\RequirePackage{tikz}

\documentclass[default,iicol]{sn-jnl}
\usepackage{todonotes}

\usepackage{balance}
\usepackage{amsthm, amsmath, amsfonts, mathtools}
\usepackage{thmtools}
\usepackage{thm-restate}
\usepackage{xspace}
\usepackage[capitalize]{cleveref}
\crefname{enumi}{Statement}{Statements}
\usepackage{tikz}

\newcommand{\eps}{\varepsilon}
\newcommand{\Oh}{\operatorname{O}}
\newcommand{\E}{\mathbb{E}}
\def\protocol#1{\textsc{#1}\xspace}
\def\norm#1{\lVert#1\rVert}
\def\abs#1{\lvert#1\rvert}

\newcommand{\HKS}[1][{}]{(G=(V,E),\eps, x_{#1})}
\newcommand{\posV}[3][{}]{x_{#3}(#2)}
\newcommand{\moveV}[2][{}]{m_{#1}(#2)}
\newcommand{\Nv}[2][{}]{\mathcal{N}_{#1}(#2)}
\newcommand{\CG}[1]{I_{{#1}}}
\newcommand{\Et}[1]{\mathcal{E}_{#1}}
\newcommand{\stateHKS}[1][{}]{(G=(V,E),\eps, x_{#1})}
\newcommand{\posVProj}[3][{}]{\bar{x}_{#3}(#2)}
\newcommand{\moveVProj}[2][{}]{\bar{m}_{#1}(#2)}
\newcommand{\infNetw}[1]{I_{#1} = (V,\Et{#1})}
\newcommand{\NvProj}[2][{}]{\bar{\mathcal{N}}_{#1}(#2)}
\newcommand{\stateHKSProj}[1][{}]{(\graphProj{#1},\eps, \bar{x}_{#1})}
\newcommand{\graphProj}[1]{\bar{G}=(V,\bar{E}_{#1})}

\jyear{2022}

\theoremstyle{thmstyleone}
\newtheorem{theorem}{Theorem}

\newtheorem{lemma}[theorem]{Lemma}
\newtheorem{corollary}[theorem]{Corollary}

\theoremstyle{thmstyletwo}

\theoremstyle{thmstylethree}

\begin{document}

\title[Asynchronous Social HKS]{Asynchronous Opinion Dynamics in Social Networks}

\author[1]{\fnm{Petra} \sur{Berenbrink}}\email{petra.berenbrink@uni-hamburg.de}

\author[2]{\fnm{Martin} \sur{Hoefer}}\email{mhoefer@em.uni-frankfurt.de}

\author[1]{\fnm{Dominik} \sur{Kaaser}}\email{dominik.kaaser@uni-hamburg.de}

\author[3]{\fnm{Pascal} \sur{Lenzner}}\email{pascal.lenzner@hpi.de}

\author[1]{\fnm{Malin} \sur{Rau}}\email{malin.rau@uni-hamburg.de}

\author[4]{\fnm{Daniel} \sur{Schmand}}\email{schmand@uni-bremen.de}

\affil[1]{\orgdiv{Department of Informatics}, \orgname{Universität Hamburg}, \orgaddress{\street{Vogt-K\"olln-Straße 30}, \postcode{22527} \city{Hamburg}, \country{Germany}}}

\affil[2]{\orgdiv{Department of Computer Science}, \orgname{Goethe University Frankfurt}, \orgaddress{\street{Robert-Mayer-Straße 11-15}, \postcode{60325} \city{Frankfurt am Main}, \country{Germany}}}

\affil[3]{\orgname{Hasso Plattner Institute}, \orgaddress{\street{Prof.-Dr.-Helmert-Str. 2-3}, \postcode{14482} \city{Potsdam}, \country{Germany}}}

\affil[4]{\orgdiv{Center for Industrial Mathematics}, \orgname{University of Bremen}, \orgaddress{\street{Bibliothekstraße 5}, \postcode{28359} \city{Bremen}, \country{Germany}}}

\abstract{Opinion spreading in a society decides the fate of elections, the success of products, and the impact of political or social movements.
A prominent model to study opinion formation processes is due to Hegselmann and Krause. It has the distinguishing feature that stable states do not necessarily show consensus, i.e., the population of agents might not agree on the same opinion.
 
We focus on the social variant of the Hegselmann-Krause model. There are $n$ agents, which are connected by a social network. Their opinions evolve in an iterative, asynchronous process, in which agents are activated one after another at random. When activated, an agent adopts the average of the opinions of its neighbors having a similar opinion (where similarity of opinions is defined using a parameter $\eps$). Thus, the set of influencing neighbors of an agent may change over time.

We show that such opinion dynamics are guaranteed to converge for any social network. We provide an upper bound of $\Oh(n\lvert E\rvert ^2 (\eps/\delta)^2)$ on the expected number of opinion updates until convergence to a stable state, where $\lvert E\rvert$ is the number of edges of the social network, and $\delta$ is a parameter of the stability concept. For the complete social network we show a bound of $\Oh(n^3(n^2 +  (\eps/\delta)^2))$ that represents a major improvement over the previously best upper bound of $\Oh(n^9 (\eps/\delta)^2)$. }


\maketitle

\section{Introduction}

Our opinions are not static. On the contrary, opinions are susceptible to dynamic changes, and this is heavily exploited by (social) media, influencers, politicians, and professionals for public relations campaigns and advertising. The way we form our opinions is not a solitary act that simply combines our personal experiences with information from the media. Instead, it is largely driven by interactions with our peers in our social network. We care about the opinions of our peers and relatives, and their opinions significantly influence our own opinion in an asynchronous dynamic process over time. Such opinion dynamics are pervasive in many real-world settings, ranging from small scale townhall meetings, community referendum campaigns, parliamentary committees, and boards of enterprises to large scale settings like political campaigns in democratic societies or peer interactions via online social networks.  

The aim for understanding how opinions are formed and how they evolve in multi-agent systems is the driving force behind an interdisciplinary research effort in diverse areas such as sociology, economics, political science, mathematics, physics, and computer science. Initial work on these issues dates back to Downs~\cite{downs1957economic} and early agent-based opinion formation models as proposed by Abelson and Bernstein~\cite{abelson1963computer}.

In this paper we study an agent-based model for opinion formation on a social network where the opinion of an agent depends both on its own intrinsic opinion and on the opinions of its network neighbors. One of the earliest influential models in this direction was defined by \mbox{DeGroot}~\cite{DeGroot74}. In this model the opinion of an agent is iteratively updated to the weighted average of the opinions of its neighbors. Later, Friedkin and Johnsen~\cite{FJ90} extended this by incorporating private opinions.
Every agent has a private opinion which does not change and an expressed opinion that changes over time. The expressed opinion of an agent is determined as a function of the expressed opinions of its neighbors and its private opinion. 

The main focus of our paper is the very influential model by Hegselmann and Krause~\cite{HK02} that adds an important feature: the set of neighbors that influence a given agent is no longer fixed, and the agents' opinions and their respective sets of influencing neighbors co-evolve over time. At any point in time the set of influencing neighbors of an agent are all the neighbors in a given static social network with an opinion close to their own opinion. Hence, agents only adapt their opinions to neighboring agents having an opinion that is not too far away from their own opinion. Note that this adaption, in turn, might lead to a new set of influencing neighbors. In sociology this wide-spread behavior is known as homophily~\cite{McPherson01}, which, for example, governs the formation of social networks and explains residential segregation. Co-evolutionary opinion formation helps to analyze and explain current phenomena like filter bubbles in the Internet~\cite{Pariser11} and social media echo chambers~\cite{Cinelli21} that inhibit opinion exchange and amplify extreme views. The co-evolution of opinions and the sets of influencing neighbors is the key feature of a Hegselmann-Krause system (HKS). It is also the main reason why the analysis of the dynamic behavior of a HKS is highly non-trivial and challenging.

Typical questions studied are the convergence properties of the opinion dynamics: Is convergence to stable states guaranteed, and if yes, what are upper and lower bounds on the convergence time? Guaranteed convergence is essential since otherwise the predictive power of the model is severely limited. Moreover, studying the convergence \emph{time} of opinion dynamics is crucially important. In general, the analysis of stable states is significantly more meaningful if these states are likely to be reached in a reasonable amount of time, i.e., if quick convergence towards such states is guaranteed. If systems do not stabilize in a reasonable time, stable states lack justification as a prediction of the system's behavior. 

Researchers have investigated the convergence to stable states and the corresponding convergence speed in many variants of the Hegselmann-Krause model. The existing work can be categorized along two dimensions: complete or arbitrary social network and \emph{synchronous} or \emph{asynchronous} updates of the opinions.
Synchronous opinion updates means that \emph{all} agents update their opinion at the same time. In systems with asynchronous updates a single agent is selected uniformly at random and only this agent updates its opinion. While the main body of recent work focuses on HKSs assuming the complete graph as social network and the synchronous update rule, empirical simulations have also been performed with asynchronous updates on arbitrary social networks. Interestingly, convergence guarantees and convergence times for the latter case are, to the best of our knowledge, absent from the literature so far. This case is arguably the most realistic setting as social networks are typically sparse, i.e.,  non-complete, and social interactions and thereby opinion exchange usually happens in an uncoordinated asynchronous fashion.

In this paper we study the following \emph{Hegselmann-Krause system (HKS)}. We have $n$ agents and their opinions are modeled by points in $d$-dimensional Euclidean space $\mathbb{R}^d$, for some $d\geq 1$. The agents are connected by a \emph{social network} which does not change over time. At any point of time the set of \emph{influencing neighbors} of an agent is the subset of its neighbors (in the social network) with an opinion of distance at most $\eps > 0$  from its own opinion. 
We assume that in each step a random agent is activated and its opinion is updated to the average of its current opinion and the opinion of all current influencing neighbors. Note in such an asynchronous HKS stable states in the sense that no agent will change its opinion might never be reached. This can be seen by a simple example with two nodes and one edge. Hence, we adopt a natural stability criterion defined by Bhattacharyya and Shiragur~\cite{BhattacharyyaS15}.
A HKS is in a \emph{$\delta$-stable state} if and only if each edge in the influence network has length at most $\delta$. For this scenario we prove that the convergence of the opinion dynamics is guaranteed. We give an upper bound on the expected convergence time of  
\[
\Oh(n\abs{E}^2 (\eps/\delta)^2) \leq \Oh(n^5 (\eps/\delta)^2),
\]
where $\abs{E}$ is the cardinality of the edge set of the given social network. We demonstrate the tightness of our derived upper bound by providing analytical lower bounds as well as empirical simulations for several topologies of the underlying social network topologies.
Note that for complete graphs as social network our bound of $\Oh(n^3(n^2 + (\eps/\delta)^2))$ improves the best previously known upper bound of $\Oh(n^9 (\eps/\delta)^2)$~\cite{EtesamiB15}. 

\subsection{Related Work} 
We focus our discussion on recent research on Hegselmann-Krause systems and other opinion formation models.

\paragraph{Synchronous HKSs on Complete Networks}
Most recent research focused on  synchronous opinion updates in complete social networks. 
For this setting it is known that the process always converges to a state where no agent changes its opinion anymore~\cite{Chazelle11}. We denote such states as \emph{perfectly stable states}. 
Touri and Nedic~\cite{TouriN11} prove that any one-dimensional HKS  converges in $\Oh(n^4)$ synchronous update rounds to a perfectly stable state.
Bhattacharyya et al.~\cite{BhattacharyyaBCN13} improve this upper bound to $\Oh(n^3)$. For $d$ dimensions they show a  convergence time of $\Oh(n^{10} d^2)$.
For arbitrary $d$ Etesami and Başar~\cite{EtesamiB15} establish a bound of $\Oh(n^6)$ rounds, which is independent of the dimension $d$.
Finally, Martinsson~\cite{Martinsson15} shows that any synchronous $d$-dimensional HKS  converges within $\Oh(n^4)$ update rounds to a perfectly stable state.

Regarding lower bounds, Bhattacharyya et al.~\cite{BhattacharyyaBCN13} construct two-dimensional instances that need at least $\Omega(n^2)$ update rounds before a perfectly stable state is reached.
Later, Wedin and Hegarty~\cite{WedinH15} show that this lower bound holds even in one-dimensional systems.

\paragraph{Synchronous HKSs on Arbitrary Social Networks} 
In \cite{parasnis2019convergence}, the authors use the probabilistic method to prove that 
the expected convergence time to a perfectly stable state is infinite for general networks.
This also holds for a slightly weaker stability concept than perfect stability: in all future steps an agent's opinion will not move further than by a given distance $\delta$. To show their result the authors construct a HKS with infinitely many oscillating states. 
Their stability notion is also different to the one considered in this paper. We analyze the time to reach a $\delta$-stable state which is defined as a state where any edge in the influence network has length at most $\delta$ (see \cref{sec:model}).
For $\delta$-stability  Bhattacharyya and Shiragur~\cite{BhattacharyyaS15} prove that a synchronous HKS with an arbitrary social network reaches a $\delta$-stable state in $\Oh(n^5 (\eps/\delta)^2)$ synchronous rounds. 

\paragraph{Asynchronous HKSs} 
Compared to the synchronous case, the existing results for asynchronous HKSs are rather limited. On the empirical side, Fortunato~\cite{fortunato2005consensus} investigated the consensus threshold with uniformly chosen initial opinions in asynchronous dynamics on non-complete social networks like grids, Erd\H{o}s-Rényi graphs, or scale-free random graphs.
To the best of our knowledge, convergence guarantees and convergence times on non-complete networks were first studied by Etesami and Başar~\cite{EtesamiB15} where the authors consider \emph{$\delta$-equilibra} in contrast to $\delta$-stable states. They define a $\delta$-equilibrium as a state where each connected component of the influence network has an Euclidean diameter of at most $\delta$ and prove that the expected number of update steps to reach such a state is bounded by $\Oh(n^9 (\eps/\delta)^2)$ for the complete social network. In general, $\delta$-equilibria are a proper subset of the set of $\delta$-stable states. However, in \cref{sec:MainResults} we discuss the equivalence of both stability notions on complete social networks.

\paragraph{Other Opinion Formation Models}
In the seminal models by Friedkin and Johnsen~\cite{FJ90} (extending earlier work by DeGroot~\cite{DeGroot74}) each agent has an innate opinion and strategically selects an expressed opinion that is a compromise of its innate opinion and the opinions of its neighbors. Recently, co-evolutionary and game-theoretic variants were studied~\cite{DBLP:journals/geb/BindelKO15,BGM13,BFM18,EpitropouFHS19,FotakisPS16}, and the results focus on equilibrium existence and social quality, measured by the price of anarchy.
In the AI and multi-agent systems community, opinion formation is studied intensively. In~\cite{DBLP:conf/aaai/AulettaFF19} a co-evolutionary model is investigated, where also the innate opinion may change over time.
There is also substantial work on understanding opinion diffusion, i.e., the process of how opinions spread in a social network~\cite{DBLP:conf/ijcai/BredereckE17,DBLP:conf/ijcai/BredereckJK20,DBLP:conf/atal/BotanGP19,DBLP:conf/ijcai/Anagnostopoulos20,DBLP:conf/ijcai/FaliszewskiGKT18,DBLP:conf/atal/DeBG18}.
Moreover, in~\cite{DBLP:conf/atal/CoatesHK18,DBLP:conf/atal/CoatesHK18a} a framework and a simulator for agent-based opinion formation models is presented.
Opinion dynamics and in particular the emergence of echo chambers is modeled with tools from statistical physics in~\cite{fu2008coevolutionary,evans2018opinion}

Another line of related research on opinion dynamics has its roots in randomized rumor spreading and distributed consensus processes (see~\cite{DBLP:journals/sigact/BecchettiCN20} for a rather recent survey).
Communication in these models is typically restricted to constantly many neighbors.
A simple and natural protocol in this context is the \protocol{Voter} process \cite{DBLP:journals/iandc/HassinP01,DBLP:journals/networks/NakataIY00,DBLP:conf/podc/CooperEOR12,DBLP:conf/icalp/BerenbrinkGKM16}, where every agent adopts in each round the opinion of a single, randomly chosen neighbor.
Similar processes are the \protocol{TwoChoices} process~\cite{DBLP:conf/icalp/CooperER14,DBLP:conf/wdag/CooperERRS15,DBLP:conf/wdag/CooperRRS17}, the \protocol{3Majority} dynamics~\cite{DBLP:journals/dc/BecchettiCNPST17, DBLP:conf/podc/GhaffariL18, DBLP:conf/podc/BerenbrinkCEKMN17}, and the Undecided State Dynamics~\cite{DBLP:journals/dc/AngluinAE08,DBLP:conf/soda/BecchettiCNPS15,DBLP:conf/mfcs/ClementiGGNPS18,DBLP:conf/podc/GhaffariP16a, DBLP:conf/icalp/BerenbrinkFGK16,BankhamerSODA22}.

\subsection{Model and Notation} \label{sec:model}
A \emph{Hegselmann-Krause system (HKS)} $\HKS$ in $d$ dimensions is defined as follows. We are given a \emph{social network} $G = (V,E)$ and a \emph{confidence bound} $\eps \in \mathbb{R}_+$.
The $n$ nodes of the social network correspond to the agents, and each agent $v \in V$ has an \emph{initial opinion $\posV{v}{ } \in \mathbb{R}^d$}. We will use the terms agents and nodes interchangeably.
As the opinion of agent $v$ is represented by a point in the $d$-dimensional Euclidean space, we sometimes call it \emph{the position of $v$}. 
In step $t \in \mathbb{Z}_{\geq 0}$ the opinion of agent $v \in V$ is denoted as $\posV{v}{t} \in \mathbb{R}^d$, where $\posV{v}{0} = \posV{v}{ }$. 
For some constant confidence bound $\eps\in \mathbb{R}_+$ we define the \emph{influencing neighborhood} of agent $v \in V$ at time $t$ as 
\begin{align*}
    &\Nv[t]{v} = \\
    &\{v\} \cup \{u \in V \mid \{u,v\} \in E, \|\posV{u}{t}-\posV{v}{t}\|_2 \leq \eps \}\;.
\end{align*}
In each step $t$ one agent $v \in V$ is chosen uniformly at random and updates its position according to the rule
\[
\posV{v}{t+1} = \frac{\sum_{u \in\Nv[t]{v} }\posV{u}{t}}{\lvert\Nv[t]{v}\rvert}\;.
\] If $\posV{v}{t} \neq \posV{v}{t+1}$, then we say that (the opinion of) agent $v$ has \emph{moved}.
Also, in an update of agent $v$'s position in step $t$, all other agents do not change their positions, i.e., $\posV{u}{t+1} = \posV{u}{t}$ for $u \neq v$.

Given a social network $G = (V,E)$, we define for any edge $e = \{u,v\} \in E$ at time $t$ the 
length of $e$ as $\|\posV{e}{t}\|_2 = \| \posV{u}{t}-\posV{w}{t}\|_2$.
We define the \emph{available movement} $\moveV[t]{v}$ of agent $v \in V$ at time $t$ as the $d$-dimensional vector
\[
\moveV[t]{v} = \sum_{u \in \Nv[t]{v}} \frac{\posV{u}{t}-\posV{v}{t}}{\lvert\Nv[t]{v}\rvert}\;.
\]
Note that $\moveV[t]{v} = \posV{v}{t+1} - \posV{v}{t}$ if $v$ is chosen in step $t$, and hence $\|\moveV[t]{v}\|_2$ denotes the distance the agent moves when activated in step $t$.
The \emph{influence network} $\CG{t}$ in step $t$ is given by the social network $G$ restricted to edges that have length at most $\eps$.
More formally, it is defined as $\CG{t} = (V,\Et{t})$, where $e = \{u,v\} \in \Et{t}$ if and only if $u \in \Nv[t]{v}$, i.e., $\|\posV{e}{t}\|_2 \leq \eps$. 
We define the \emph{state} of a HKS $\HKS$ at time $t$ as $S_t = \stateHKS[t]$ and it refers to the positions of the agents at that specific time.
If clear from the context, we omit the parameter $t$.
For a fixed state $S$, the term $\Nv{v}$ denotes the influencing neighborhood in this state.

We are interested in the expected number of steps that are required until the HKS reaches a $\delta$-stable state, which is 
a natural stability criterion defined by Bhattacharyya and Shiragur~\cite{BhattacharyyaS15}. A HKS is in a \emph{$\delta$-stable state} if and only if each edge in the influence network has length at most $\delta$. Intuitively, in such a state each agent has a small incentive to further revise the opinion. Hence, it is reasonable to assume that such states represent a stable configuration of the system. Strictly speaking, however, in a $\delta$-stable state the HKS might not be stabilized \emph{entirely} in the sense that agents are unable to achieve further improvement by a deviation \emph{at all}. If agents would continue to revise their opinions, the HKS might subsequently be able to leave the $\delta$-stable state. Put differently, not all such states are attractive. We note, however, that this is a condition shared by the vast majority of approximate stability or equilibrium concepts defined in the literature.

We call the number of steps to reach a $\delta$-stable state the \emph{convergence time} of the system. To track the progress towards convergence, we define the following potential function
for any state $S=\stateHKS$ of a $d$-dimensional HKS $\HKS$:
\[
\Phi(S) = \sum_{\{u,v\}\in E} \min\{ \|\posV{u}{} - \posV{v}{}\|_2^2, \eps^2\} .
\]
This potential is upper-bounded by $\Phi(S) \leq \abs{E} \eps^2$.

\subsection{Our Contribution}
We study the convergence time to a $\delta$-stable state in Hegselmann-Krause systems with an arbitrary initial state and an arbitrary given social network, where we update one uniformly at random chosen agent in each step. To the best of our knowledge, this is the first analysis of the variant of HKSs that feature asynchronous opinion updates on a given arbitrary social network. For these systems, we prove the following:

\begin{restatable}{theorem}{UpperBoundArbitraryGraph}
\label{thm:upper-bound}
For a $d$-dimensional HKS $S_0 = \HKS$, the expected convergence time to a $\delta$-stable state under uniform random asynchronous updates is $\Oh(\Phi(S_0)n\abs{E}/\delta^2)\leq\Oh(n \abs{E}^2 \left(\eps/\delta \right)^2)$.
\end{restatable}

For graphs with $\abs{E} = \Oh(n)$, for example graphs with constant maximum node degree, the theorem immediately shows an expected convergence time of $\Oh(n^3\left(\eps/\delta\right)^2)$. Interestingly, our upper bound on the expected convergence time in the asynchronous process on arbitrary social networks is of the same order as the best known upper bound of $\Oh(n^5\left(\eps/\delta \right)^2)$ for the synchronous process~\cite{BhattacharyyaS15} where \emph{all} agents are activated in parallel. 

Furthermore, we show that the convergence time stated in \cref{thm:upper-bound} also transfers to the model of Etesami and Başar~\cite{EtesamiB15}. They showed that a HKS with asynchronous opinion updates on a complete social network converges to a \emph{$\delta$-equilibrium} in $\Oh(n^9\left(\eps/\delta \right)^2)$ steps, thus it is a major improvement over their analysis. However, since on arbitrary social networks $\delta$-stability does not imply a $\delta$-equilibrium, it is open if the bound given in Theorem~\ref{thm:upper-bound} also holds for the convergence time to $\delta$-equilibria.

Moreover, for the special case of a complete social network with asynchronous opinion updates, i.e., the case considered by Etesami and Başar~\cite{EtesamiB15}, we show the following even stronger result that holds for arbitrary $\delta$:

\begin{restatable}{theorem}{UpperBoundComplete}
\label{UpperBoundComplete}
Let $\HKS$ be any instance of a $d$-dimen\-sional HKS and let $G=K_n$ be the complete social network. 
Using uniform random asynchronous update steps, the expected convergence time to a 
$\delta$-stable state is at most $\Oh\left( n^3\left(n^2 +(\eps/\delta)^2\right)\right)$.
\end{restatable}

To prove these results, we extend the potential function used in~\cite{EtesamiB15}.
The main ingredient for strongly improving the upper bound derived in \cite{EtesamiB15} is to significantly tighten and generalize their proof. To do so, we develop a projection argument (see \cref{lem:project}) and a new analysis of the expected available movement of a randomly chosen agent. This allows us to improve the bound on the expected drop of the potential function (see \cref{lma:ExpPotentialDrop}).
 
To complement our upper bound results, we demonstrate that our analysis method is tight in the sense that by using this potential function and studying the step-by-step drop, one cannot improve the results. 
We present a family of instances and initial states where the expected potential drop is exactly of the same order as our upper bound (see \cref{rem:RemakTightniss}).  
Moreover, we present a family of one-dimensional HKSs and initial states where  $\Omega(\Phi(S_0)n\abs{E}/\eps^2)$ steps are needed to reach a $\delta$-stable state (see \cref{thm:lower-bound}), thereby matching the upper bound shown in Theorem~\ref{thm:upper-bound} in terms of the order of $n$.

\begin{restatable}{theorem}{thmLowerBound}
\label{thm:lower-bound}
For $\eps= 1$ and $\delta < 1/2$, there exists a family of social 1-dimensional HKSs $(S_{0,4n})_{n \in \mathbb{N}}$ where any given update sequence needs at least $\Omega(\Phi(S_{0,4n})n \abs{E}/\eps^2)=\Omega(n^4)$ steps to reach a $\delta$-stable state.
\end{restatable}

Notably, this lower bound applies for \emph{arbitrary update sequences}, while our upper bound applies when the updating agent is chosen uniformly at random. Thus, even when resorting to a smarter choice of the updating agent, one cannot drastically reduce the convergence time in the worst case.

Last but not least, in \cref{sec:simulations} we provide some simulation results for two specific social network topologies. Our empirically derived lower bounds asymptotically match our theoretically proven upper bound from \cref{thm:upper-bound}.

\section{Social Hegselmann-Krause Systems}
\label{sec:MainResults}
In this section we prove \cref{thm:upper-bound} in three steps. 
Recall that for a HKS in $d$ dimensions the opinion are represented by points in the $d$-dimensional Euclidean space. 
First, in \cref{lem:project} we show that for each HKS in $d$ dimensions there exists a mapping to a suitable $1$-dimensional HKS, such that the length of all edges does not increase, and the influence network (consisting of the active edges) as well as the length of the longest edge $\lambda$ is preserved. 
We use this projection in the second step (see \cref{lem:min_movement}) where we only consider HKS in one dimension: 
We prove that $\sum_{v \in V}(\lvert\Nv[t]{v}\rvert\cdot \|\moveV[t]{v}\|_2) \geq 2\lambda$, where $\Nv[t]{v}$  is the set of neighbors in the influence network and $\moveV[t]{v}$ is the available movement of the node.
In the third step, we prove that the potential drop due to activating an agent $v$ can be lower bounded by $(\abs{\Nv[t]{v}}+1) \cdot \|\moveV[t]{v}\|_2^2$ (see \cref{lem:potentialDrop}).
Finally, in \cref{lma:ExpPotentialDrop} we combine these three insights to bound the potential drop.

Let $S  = \stateHKS$ be a state of some $d$-dimensional HKS with influence network $\infNetw{}$.
For some arbitrary edge $e=\{u,w\} \in \Et{}$, we will project the state $S$ to a state $\bar{S}_e$ of some $1$-dimensional HKS.
We define the projected state $\bar{S}_e$ along edge $e=\{u,w\}$ with the help of the projection vector \[p = \frac{(\posV{u}{} - \posV{w}{})}{\|\posV{u}{} - \posV{w}{}\|_2}\;,\] where the order of $u$ and $w$ is chosen arbitrarily.
We define
\[\bar{S}_e = \stateHKSProj\;,\] 
as follows. 
We project the position of each agent $v \in V$ to
\[\posVProj{v}{} = \posV{v}{}^\top p \in \mathbb{R}\;.\]
Furthermore, in the graph $\graphProj{}$ of the projected system, we restrict the set of edges $\bar{E}$ to the ones, which are edges of the influence network in the original state, i.e.,
$\bar{E} = \Et{}$.
For an agent $v \in V$, we denote by $\NvProj{v}$ its influencing neighborhood, and by $\moveVProj{v}$ its available movement in $\bar{S}_e$.

In the following lemma, we prove that the projected system behaves similarly to the original system in the sense that the length of the edge $e$ stays the same and the influence network does not change. 
Furthermore, the agents in the original HKS move at least as much as the agents in the projected state, when activated.

\begin{lemma}
\label{lem:project}
Let $S = \stateHKS$ be a state of a $d$-dimensional HKS with influence network $\infNetw{}$
and $e = \{u,w\} \in \Et{}$.
Then for any $v,v' \in V$ and the projected state $\bar{S}_e$ defined as above it holds that

\begin{align}
\|\posV{u}{}-\posV{w}{}\|_2 &= \lvert\posVProj{u}{} - \posVProj{w}{}\rvert\;, \label{en:lem:project:1} \\
\|\posV{v}{}-\posV{v'}{}\|_2 & \geq  \lvert\posVProj{v}{} - \posVProj{v'}{}\rvert\;, \label{en:lem:project:3} \\
\Nv{v} &= \NvProj{v}\;,\label{en:lem:project:2} \text{ and}  \\
\sum_{v \in V} (\lvert\Nv{v}\rvert\cdot\|\moveV[]{v}\|_2) &\geq \sum_{v \in V} (\lvert\NvProj{v}\rvert\lvert\moveVProj[]{v}\rvert) \;.\label{en:lem:project:4}
\end{align}
\end{lemma}

\begin{proof}
Let $p$ be the projection vector used to generate $\bar{S}_e$.
To see statement \eqref{en:lem:project:1} note that 
\begin{align*}
    &\lvert\posVProj{u}{} - \posVProj{w}{}\rvert \\
    &= \left\lvert\posV{u}{}^\top p - \posV{w}{}^\top p \right\rvert = \left\lvert(\posV{u}{}-\posV{w}{})^T p\right\rvert\\
    &=\left\lvert \frac{(\posV{u}{} -\posV{w}{})^\top (\posV{u}{} -\posV{w}{})}{ \norm{(\posV{u}{} -\posV{w}{})}_2}\right\rvert\\
    &= \norm{(\posV{u}{} -\posV{w}{})}_2.
\end{align*}
To prove statement \eqref{en:lem:project:3}, we show that for each pair $v,v' \in V$ it holds that $\|\posV{v}{}-\posV{v'}{}\|_2 \geq \lvert\posVProj{v}{}-\posVProj{v'}{}\rvert$:
\begin{align*}
    &\lvert\posVProj{v}{}-\posVProj{v'}{}\rvert \\
    &= \left\lvert\posV{v}{}^\top p- \posV{v'}{}^\top p\right\rvert
    =  \left\lvert(\posV{v}{} -\posV{v'}{})^\top p\right\rvert\\
    &= \left\lvert\frac{(\posV{v}{} -\posV{v'}{})^\top (\posV{u}{} - \posV{w}{})}{ \norm{(\posV{u}{} -\posV{w}{})}_2}\right\rvert\\
    &\overset{\text{C.-S.}}{\leq} \norm{(\posV{v}{} - \posV{v'}{})}_2\;.
\end{align*}
The last inequality uses the Cauchy-Schwarz inequality (C.-S.).

To see \eqref{en:lem:project:2}, note that (because of \eqref{en:lem:project:3}) the difference between projected positions of agents is at most as large as the difference between their original positions. Since $\bar{E}$ contains only the edges of the influence network in the original state, it holds that $\Nv{v} = \NvProj{v}$.

Finally, it holds that 
\begin{align*}
&\|\moveV{v}\|_2\\
&= \left\|\frac{\sum_{u \in \Nv{v}} (\posV{u}{} - \posV{v}{})}{\abs{\Nv{v}}} \right\|_2 \\
&= \frac{\left\|\sum_{u \in \Nv{v}} (\posV{u}{} - \posV{v}{})\right\|_2}{\abs{\Nv{v}}}  \\
&\overset{\text{C.-S.}}{\geq} \frac{\left\lvert\left(\sum_{u \in \Nv{v}} (\posV{u}{} - \posV{v}{})\right)^\top (\posV{u}{} - \posV{w}{})\right\rvert}{\abs{\Nv{v}}\|\posV{u}{} - \posV{w}{}\|_2}\\
&= \left\lvert\frac{\left(\sum_{u \in \Nv{v}} (\posV{u}{}^\top p - \posV{v}{}^\top p)\right)}{\abs{\Nv{v}}}\right\rvert  \\
&= \left\lvert \frac{\sum_{j \in \NvProj{v}}(\posVProj{u}{} - \posVProj{v}{})}{\abs{\NvProj{v}}} \right\rvert = \lvert\moveVProj{v}\rvert\;,
\end{align*}
and hence 
\begin{align*}
\sum_{v \in V} \lvert\Nv{v}\rvert\|\moveV{v}\|_2 
&\geq 
\sum_{v \in V} \abs{\NvProj{v}}\lvert\moveVProj{v}\rvert\;.
\end{align*}
\end{proof}
\noindent We now prove a lower bound on the total available movement of agents.

\begin{lemma}
\label{lem:summedMovement}
Let $S =\stateHKS$ be a state of a $1$-dimensional HKS, let $c\in \mathbb{R}$ and $V_\ell = \{v \in V \mid \posV{v}{} \leq c\}$ and $V_r = V \setminus V_\ell$.
Define $E_{\ell,r} = \{\{u,w\} \in E_I \mid u \in V_l, w \in V_r\}$.
Then it holds that
\[\sum_{v \in V}\abs{\Nv{v}}\abs{\moveV{v}} \geq 2 \sum_{e \in E_{\ell,r}} \|\posV{e}{}\|_2 \]
\end{lemma}
\begin{proof}
We observe
\begin{align*}
\MoveEqLeft \sum_{v \in V_\ell} \abs{\Nv{v}}\abs{\moveV{v}} \geq \sum_{v \in V_\ell} \abs{\Nv{v}}\moveV{v}\\ 
=& \sum_{v \in V_\ell} \abs{\Nv{v}} \sum_{u \in \Nv{v}} \frac{\posV{u}{} -\posV{v}{}}{\abs{\Nv{v}}} \\
=& \sum_{v \in V_\ell}\sum_{u \in V_r \cap \Nv{v}} (\posV{u}{} -\posV{v}{})\\
&+ \sum_{v \in V_\ell}\sum_{u \in V_\ell \cap \Nv{v}} (\posV{u}{} -\posV{v}{})\;.
\end{align*}
Note that for each edge $e = \{v, u\}$ with $v,u \in V_{\ell}$ the second sum contains $x(u) -x(v)$ as well as $x(v) -x(u)$. As such,
\[
\sum_{v \in V_\ell}\sum_{u \in V_\ell \cap \Nv{v}} (\posV{u}{} -\posV{v}{}) = 0
\]
Furthermore, for each edge $e = \{v, u\}$ with $v \in V_{\ell}$ and $u \in V_r$ it holds that $x(u)-x(v) > 0$, since $x(u) > 0$ and $x(v) < 0$.
As a consequence,
\begin{align*}
\MoveEqLeft \sum_{v \in V_\ell} \abs{\Nv{v}}\abs{\moveV{v}} \geq \sum_{v \in V_\ell}\sum_{u \in V_r \cap \Nv{v}} (\posV{u}{} -\posV{v}{})\\
=& \sum_{e \in E_{\ell,r}} \|\posV{e}{}\|_2\;.
\end{align*}
Similarly, it holds that
\begin{align*}
\sum_{v \in V_r} \abs{\Nv{v}}\abs{\moveV{v}} &\geq \left\lvert \sum_{v \in V_r} \abs{\Nv{v}}\moveV{v} \right\rvert\\
&= \sum_{e \in E_{\ell,r}} \|\posV{e}{}\|_2\;.
\end{align*}
The lemma follows by combining the two results as $V_r = V \setminus V_{\ell}$.
\end{proof}

\begin{corollary}
\label{lem:min_movement}
Let $S = \stateHKS$ be a state of a $d$-dimensional HK system. 
Let $\lambda$ be the length of a longest edge in the influence network.
Then
\[
\sum_{v \in V}\abs{\Nv[t]{v}}\cdot \|\moveV[t]{v}\|_2 \geq 2 \lambda .
\]
\end{corollary}

\begin{proof}
Let edge $e = \{ v_\ell,v_r \} \in \mathcal{E}$ be a longest edge in the influence network and $\|\posV{e}{}\|_2 = \lambda$.
Let $\bar{S}_e = \stateHKSProj$ be the state projected to one dimension along the edge $e$. 
By \cref{lem:project} \cref{en:lem:project:4}, we know that 
\[
\sum_{v \in V}\abs{\Nv{v}}\cdot \|\moveV{v}\|_2 \geq \sum_{v \in V}\abs{\NvProj{v}}\cdot \abs{\moveVProj{v}}\;.\]
Furthermore, by \cref{lem:project}, we know that the influence network in both systems has the same set of edges (\cref{en:lem:project:2}), the longest edge $e$ preserves its length in the projection (\cref{en:lem:project:1}), and all other edges do not increase their length (\cref{en:lem:project:3}). Therefore, the length of the longest edge in the influence network of $\bar{I}_t$ is equal to the length of the longest edge in $I_t$.
Hence $e = \{u,w\} \in E$ is a longest edge in the influence network $\bar{I}$ with $\|\posV{e}{}\|_2 = \lambda$.

Analogously to \cref{lem:summedMovement}, we partition $V$ into two sets $V_\ell$ and $V_r$ at $c = (\posV{{u}}{}+\posV{{w}}{})/2$ and define $E_{\ell,r} = \{\{v, v'\} \mid v \in V_\ell, v' \in V_r\}$.
Note that $e  \in E_{\ell,r}$ and hence
$
    \sum_{v \in V} \abs{\moveVProj{v}} \abs{\NvProj{v}} \geq   2 \sum_{e \in E_{\ell,r}} \|\posV{e}{}\|_2 \geq 2 \lambda\;.
$
\end{proof}

In the next step, we will prove a lower bound on the drop in the potential when updating any agent $v \in V$.
\begin{restatable}{lemma}{EtesamiBasar}
\label{lem:potentialDrop}
Let $S_t = \stateHKS[t]$ be the state of some $d$-dimensional HKS $\HKS$. 
Suppose we update the position of agent $v$ and $v$ moves by $\moveV[t]{v}$. 
Let \[S_{t+1} = \stateHKS[t+1]\] be the new state. The potential decreases by at least 
\[\Phi(S_t) - \Phi(S_{t+1}) \geq (\abs{\Nv[t]{v}}+1) \cdot \|\moveV[t]{v}\|_2^2.\]
If the influence network does not change from step $t$ to $t+1$, we obtain equality.
\label{lem:etesami_modified}
\end{restatable}

\begin{proof}
As we activate $v$, the position of agents $u \neq v$ does not change, but the set of active edges can change and is updated from $\Et{t}$ to $\Et{t+1}$.
To bound the potential change we consider edges $\Et{t} \cap \Et{t+1}$, $\Et{t} \setminus \Et{t+1}$ and $\Et{t+1} \setminus \Et{t}$.
In the set $\Et{t} \setminus \Et{t+1}$ the length of the edges increases above $\eps$ while in the set $\Et{t+1} \setminus \Et{t}$ the length decreases below $\eps$.

By the definition of $\Phi$, we have
\begin{align*}
    \MoveEqLeft \Phi(S_t) - \Phi(S_{t+1})\\
    \MoveEqLeft = \sum_{\{u,v\}\in E} ( \min\{ \|\posV{v}{t} - \posV{u}{t}\|_2^2, \eps^2\}\\
    & - \min\{ \|\posV{v}{t+1} - \posV{u}{t+1}\|_2^2, \eps^2\} )\\
    \MoveEqLeft= \sum_{\substack{\{u,v\} \in \\ \Et{t} \cap \Et{t+1}}} \left(\|\posV{v}{t} - \posV{u}{t}\|_2^2 - \|\posV{v}{t+1} - \posV{u}{t+1}\|_2^2\right)\\
    & + \sum_{\substack{\{u,v\} \in \\ \Et{t}\setminus \Et{t+1}}} \left(\|\posV{v}{t} - \posV{u}{t}\|_2^2 - \eps^2\right)\\
    & + \sum_{\substack{\{u,v\} \in \\ \Et{t+1} \setminus \Et{t}}} \left(\eps^2 - \|\posV{v}{t+1} - \posV{u}{t+1}\|_2^2\right)\\
    \MoveEqLeft\geq \sum_{\substack{\{u,v\} \in \\ \Et{t} \cap \Et{t+1}}} \left(\|\posV{v}{t} - \posV{u}{t}\|_2^2 - \|\posV{v}{t+1} - \posV{u}{t+1}\|_2^2\right)\\
    & + \sum_{\substack{\{u,v\} \in \\ \Et{t}\setminus \Et{t+1}}} \big(\|\posV{v}{t} - \posV{u}{t}\|_2^2\\
    &\quad\quad\quad\quad- \|\posV{v}{t+1} - \posV{u}{t+1}\|_2^2\big).\\
\end{align*}
Note that in this step, we have equality if $\Et{t} = \Et{t+1}$. 
We conclude
\begin{align*}
\MoveEqLeft \Phi(S_t) - \Phi(S_{t+1})\\
\MoveEqLeft \geq \sum_{\substack{\{u,v\} \in \Et{t}}} \left(\|\posV{v}{t} - \posV{u}{t}\|_2^2 - \|\posV{v}{t+1} - \posV{u}{t+1}\|_2^2\right)\\
\MoveEqLeft= \sum_{u \in\Nv[t]{v}} \left(\|\posV{v}{t} - \posV{u}{t}\|_2^2 - \|\posV{v}{t+1} - \posV{u}{t+1}\|_2^2\right)\\ 
\MoveEqLeft= \|\posV{v}{t+1} - \posV{v}{t}\|_2^2 \\
    & + \sum_{u \in\Nv[t]{v}} \big(\|\posV{v}{t} - \posV{u}{t}\|_2^2\\
    &\quad\quad\quad\quad~~ - \|\posV{v}{t+1} - \posV{u}{t}\|_2^2\big)\\
\MoveEqLeft= \|\moveV[t]{v}\|_2^2 \\
    & + \sum_{u \in\Nv[t]{v}} \big(\|\posV{v}{t} - \posV{u}{t}\|_2^2\\
    & \quad\quad\quad\quad~~- \|\posV{v}{t}+\moveV[t]{v} - \posV{u}{t}\|_2^2\big),\\
\end{align*}
Using the definition of $\|\cdot\|$, we obtain
\begin{align*}
\MoveEqLeft= \|\moveV[t]{v}\|_2^2 \\
    & +\sum_{u \in\Nv[t]{v}} (\posV{v}{t}^\top \posV{v}{t} - 2\posV{v}{t}^\top \posV{u}{t}\\
    &\quad\quad\quad+ \posV{u}{t}^\top \posV{u}{t}\\ 
    &\quad \quad \quad -(\posV{v}{t}+\moveV[t]{v})^\top (\posV{v}{t}+\moveV[t]{v})\\
    &\quad \quad \quad +2(\posV{v}{t}+\moveV[t]{v})^\top \posV{u}{t}- \posV{u}{t}^\top \posV{u}{t})\\
\MoveEqLeft= \|\moveV[t]{v}\|_2^2 \\
    & + \smashoperator{\sum_{u \in\Nv[t]{v}}} ( 
        -2\moveV[t]{v}^\top \posV{v}{t} - \moveV[t]{v}^\top \moveV[t]{v}\\
    &\quad\quad\quad+2\moveV[t]{v}^\top \posV{u}{t}) \\
\MoveEqLeft= \|\moveV[t]{v}\|_2^2 -\abs{\Nv[t]{v}} \|\moveV[t]{v}\|_2^2 \\
    & +2\moveV[t]{v}^\top \left(\sum_{u \in\Nv[t]{v}} \left(\posV{u}{t}-\posV{v}{t}\right)\right) \\ 
\MoveEqLeft= \|\moveV[t]{v}\|_2^2 -\abs{\Nv[t]{v}} \|\moveV[t]{v}\|_2^2 \\
& +2\abs{\Nv[t]{v}}\moveV[t]{v}^\top \moveV[t]{v}\\ 
\MoveEqLeft= (\abs{\Nv[t]{v}}+1)\|\moveV[t]{v}\|_2^2\;,
\end{align*}
which concludes the proof.
\end{proof}
We now have the tools to prove a lower bound on the expected potential drop in a single step.
\begin{lemma}
\label{lma:ExpPotentialDrop}
For any state $S_t=\stateHKS[t]$ of some HKS $\HKS$ in step $t$, when updating an agent chosen uniformly at random resulting in state $S_{t+1} = \stateHKS[t+1]$, the expected potential drop is at least 
\[\E[\Phi(S_t) - \Phi(S_{t+1})] \geq \frac{2 (\lambda_t)^2}{n\abs{\Et{t}}}\;,\]
where $\lambda_t$ is the length of the longest edge in the influence network $\CG{t}$ in step $t$.
\end{lemma}

\begin{proof}
From \cref{lem:etesami_modified} we know that the potential never increases: if we choose agent $v$ to be updated, 
the potential decreases by at least 
\[\Phi(S_t) - \Phi(S_{t+1}) \geq (\abs{\Nv[t]{v}}+1) \cdot \|\moveV[t]{v}\|_2^2.\] 

Let $e_t$ be a longest edge in the corresponding influence network of $S_t$.
By \cref{lem:min_movement}, we know that 
\[\sum_{v \in V}\abs{\Nv[t]{v}}\cdot \|\moveV[t]{v}\|_2 \geq 2\|e_t\|_2.\]

Using Cauchy-Schwarz $\left(\sum_{v \in V} a_v b_v \right)^2 \leq \sum_{v \in V} a_v^2 \cdot \sum_{v \in V} b_v^2$ with $a_v = \sqrt{\abs{\Nv[t]{v}}}\cdot\|\moveV[t]{v}\|_2$ and $b_v = \sqrt{\abs{\Nv[t]{v}}}$, we conclude that the expected potential drop in each step with an edge with length at least $\lambda_t$ is at least 
\begin{align*}
&\E[\Phi(S_t) - \Phi(S_{t+1})] \\
&= \sum_{v \in V}\frac{1}{n}\E[\Phi(S_t) - \Phi(S_{t+1}) \mid v \text{ is updated}]\\
&\geq \frac{1}{n}\sum_{v \in V} (\abs{\Nv[t]{v}} + 1)\|\moveV[t]{v}\|_2^2\\
&\geq \frac{1}{n}\sum_{v \in V} (\sqrt{\abs{\Nv[t]{v}}}\cdot \|\moveV[t]{v}\|_2)^2\\
&\geq \frac{1}{n}\frac{\left(\sum_{v \in V} \abs{\Nv[t]{v}}\cdot\|\moveV[t]{v}\|_2\right)^2}{\sum_{v \in V} \sqrt{\abs{\Nv[t]{v}}}^2}\\
& \geq \frac{1}{n} \cdot \frac{4 (\lambda_t)^2}{2\abs{\Et{t}}}.
\end{align*}
\end{proof}

\noindent The proof of \cref{thm:upper-bound} is a direct consequence of \cref{lma:ExpPotentialDrop}.

\UpperBoundArbitraryGraph*

\begin{proof}
Note that by definition of the potential function, we have $\Phi(S) \leq \eps^2\abs{E}$ for all states $S$.
We know by \cref{lma:ExpPotentialDrop} that the expected potential drop at any step $t$ is at least 
\[
\E[\Phi(S_t) - \Phi(S_{t+1})] \geq \frac{2 \delta^2}{n\abs{\Et{t}}} \geq \frac{2 \delta^2}{n\abs{E}}
\]
as long as there is an edge with length at least $\delta$. 

Applying the classic additive drift theorem (see, e.g.,~\cite[Theorem 2.3.1]{Lengler20} and the historic references therein) we directly observe that the expected number of steps to reach a $\delta$-stable state is upper bounded by
\[
\frac{\Phi(S)}{\frac{2 \delta^2}{n \abs{E}}} \leq \frac{\abs{E} \eps ^2}{\frac{2 \delta^2}{n \abs{E}}} = \frac{n \abs{E}^2}{2} \left(\frac{\eps}{\delta}\right)^2\;, 
\]
resulting in the bound from the theorem.
\end{proof}
\noindent Our results in \cref{thm:upper-bound} directly improve the results from Etesami and Ba\c{s}ar~\cite{EtesamiB15} even though they use a slightly different convergence criterium.
In their paper, convergence is reached if the diameter of each connected component is bounded by $\delta$, and they call this state a $\delta$-equilibrium.
They bound the expected number of update steps to reach a $\delta$-equilibrium in the complete social network by $\Oh(n^9(\eps/ \delta)^2)$.

Our result transfers to their notion of convergence as follows. Assume $\delta \leq \eps/2$ and that the length of the longest edge is at most $\delta.$
If the social network is the complete graph, each connected component in the influence network must be a complete sub-graph, see \cref{clm:conectedComponent}. 
Hence, the diameter of this connected component is also bounded by $\delta$.
Hence, if $\delta \leq \eps/2$, a $\delta$-stable state must be in $\delta$-equilibrium as well.
On the other hand, if $\delta > \eps/2$, the expected number of steps to reach a $\eps/2$-stable state and hence a $\delta$-equilibrium is bounded by $\Oh(n^5)$ by \cref{thm:upper-bound}.

\bigskip 

The next theorem shows that our bound on the potential drop per step is tight. Consequently, if we would like to improve the theorem, we have to choose a different potential function and/or consider multiple activations at once.
\begin{theorem}
\label{rem:RemakTightniss}
There is a family of instances and initial states with $\abs{E} = \Theta(n^2)$ and a potential of $\Theta(n^2\eps^2)$, where the expected potential drop is $\Theta(\eps^2/n^3)$ for the first activation. 
\end{theorem}

\begin{figure}[t]
    \centering
\resizebox{\linewidth}{!}{    
\begin{tikzpicture}
\pgfmathsetmacro{\w}{1.6}
\pgfmathsetmacro{\h}{0.6}
\pgfmathsetmacro{\x}{\w/35}
\node[circle,fill=black,inner sep=0pt,minimum size=3pt]  (la) at (0*\w,0*\h) {};
\node[circle,fill=black,inner sep=0pt,minimum size=3pt]  (lb) at (0*\w,1*\h) {};
\node[circle,fill=black,inner sep=0pt,minimum size=3pt]  (lc) at (0*\w,2*\h) {};
\node[circle,fill=black,inner sep=0pt,minimum size=3pt]  (ld) at (0*\w,3*\h) {};

\node at (0*\w,-0.5*\h) {$C_l$};

\draw[thick] (la) -- (lb);
\draw[thick] (lb) -- (lc);
\draw[thick] (lc) -- (ld);

\draw[thick] (la) to[out=120,in=-120] (lc);
\draw[thick] (lb) to[out=120,in=-120] (ld);
\draw[thick] (la) to[out=120,in=-120] (ld);

\node[circle,fill=black,inner sep=0pt,minimum size=3pt]  (l) at (\x*5,1.5*\h) {};
\node[label={[xshift=0.1*\w cm, yshift=-1*\h cm]$\ell$}] at (l) { };

\draw[thick] (la) -- (l);
\draw[thick] (lb) -- (l);
\draw[thick] (lc) -- (l);
\draw[thick] (ld) -- (l);

\node[circle,fill=black,inner sep=0pt,minimum size=3pt]  (pa) at (\w-4*\x,1.5*\h) {};
\node[circle,fill=black,inner sep=0pt,minimum size=3pt]  (pb) at (2*\w-10*\x,1.5*\h) {};
\node[circle,fill=black,inner sep=0pt,minimum size=3pt]  (pc) at (3*\w-13*\x,1.5*\h) {};
\node[circle,fill=black,inner sep=0pt,minimum size=3pt]  (pd) at (4*\w-13*\x,1.5*\h) {};
\node[circle,fill=black,inner sep=0pt,minimum size=3pt]  (pe) at (5*\w-16*\x,1.5*\h) {};
\node[circle,fill=black,inner sep=0pt,minimum size=3pt]  (pf) at (6*\w-22*\x,1.5*\h) {};

\node[circle,fill=black,inner sep=0pt,minimum size=3pt]  (r) at (7*\w-31*\x,1.5*\h) {};
\node[label={[xshift=-0.1*\w cm, yshift=-1*\h cm]$r$}] at (r) { };

\draw[thick] (l) -- node[midway,above]{$\eps - 9\hat{m}$} (pa);
\draw[thick] (pa) -- node[midway,above]{$\eps - 6\hat{m}$} (pb);
\draw[thick] (pb) -- node[midway,above]{$\eps - 3\hat{m}$} (pc);
\draw[thick] (pc) -- node[midway,above]{$\eps$} (pd);
\draw[thick] (pd) -- node[midway,above]{$\eps - 3\hat{m}$} (pe);
\draw[thick] (pe) -- node[midway,above]{$\eps - 6\hat{m}$} (pf);
\draw[thick] (pf) -- node[midway,above]{$\eps - 9\hat{m}$} (r);

\node[circle,fill=black,inner sep=0pt,minimum size=3pt]  (ra) at (7*\w-26*\x,0*\h) {};
\node[circle,fill=black,inner sep=0pt,minimum size=3pt]  (rb) at (7*\w-26*\x,1*\h) {};
\node[circle,fill=black,inner sep=0pt,minimum size=3pt]  (rc) at (7*\w-26*\x,2*\h) {};
\node[circle,fill=black,inner sep=0pt,minimum size=3pt]  (rd) at (7*\w-26*\x,3*\h) {};

\node at (7*\w-26*\x,-0.5*\h) {$C_r$};

\draw[thick] (ra) -- (rb);
\draw[thick] (rb) -- (rc);
\draw[thick] (rc) -- (rd);

\draw[thick] (ra) to[out=60,in=-60] (rc);
\draw[thick] (rb) to[out=60,in=-60] (rd);
\draw[thick] (ra) to[out=60,in=-60] (rd);

\draw[thick] (ra) -- (r);
\draw[thick] (rb) -- (r);
\draw[thick] (rc) -- (r);
\draw[thick] (rd) -- (r);

\end{tikzpicture}
}
    \caption{A state $S$ of a HKS with $\Phi(S) = \Theta(n^2 \eps)$ and an expected potential drop of $\Theta(\eps^2/n^3)$. Only edges in $\Et{0}$ are presented, and $\hat{m} = \eps/(n^2/16 +5n/4-1)$ represents the equal available movement of all nodes. Note that the state $S$ is a one-dimensional instance and the position of all nodes of the cliques $C_{\ell}$ and $C_r$ have the same position, respectively. We use the second dimension only for a better illustration of the influencing network. We call the state $S$ with its social network reduced to the edges in $\Et{0}$ a \textsl{Dumbbell} instance.} 
    \label{fig:DumbbellGraph}
\end{figure}
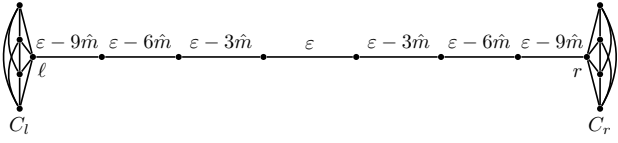

\begin{proof}
Consider the following family of $1$-dimensional HKSs $HK_n = \HKS[0]$ such that $\abs{V} = 4 n$ for any $n \in \mathbb{N}_{> 1}$, see \cref{fig:DumbbellGraph} for the example for $n = 4$.
The set of nodes $V$ is partitioned into sets $C_\ell, C_r, P, \{\ell,r\} \subseteq V$, such that $\abs{C_\ell} = \abs{C_r} = n$ and $\abs{P} = 2n -2$.
The set of edges $E$ is given such that $C_\ell$, $C_r$, and $P$ are cliques while nodes $\ell$ and $r$ are connected to all nodes.

To define the opinions of the agents that correspond to the nodes $V$ at state $S_0$, define $\hat{m} = \eps/(n^2 + 5n -1)$ and choose
\begin{itemize}
    \item $\posV{v}{0} = 0$ for each $v \in C_\ell$,
    \item $\posV{\ell}{0}= \hat{m}\cdot (n + 1)$,
    \item for each $j \in \{1,\dots, n\}$ there exists a node $v_j \in P$ with 
    \begin{align*}
    \posV[j]{v_j}{0} & = \posV[j-1]{v_{j-1}}{0} + \eps - 3 (n-j)\hat{m} 
    \end{align*}
    where we define $v_0 = \ell$.
    \item for each $j \in \{n+1,\dots, 2n-2\}$ there exists a node $v_j \in P$ with 
    \begin{align*}
    \posV[j]{v_j}{0} &= \posV[j-1]{v_{j-1}}{0}+\eps - 3(j-n)\hat{m}
    \end{align*}
    \item $\posV{r}{0} = \posV[2n-2]{v}{0} + \eps - 3 (n-1)\hat{m}$ 
    \item $\posV{v}{0} = \posV{r}{0} + (n +1)\hat{m} $ for each $v \in C_r$.
\end{itemize}
Note that all the edges inside the cliques $C_\ell \cup \{\ell\}$ and $C_r\cup \{r\}$ are in the influence network $\CG{0}$,
as well as each edge between $v_{j}$ and $v_{j+1}$ for $j \in \{0,\dots,2n-2\}$, where $v_0 = \ell$ and $v_{2n-1} = r$.
Also, 
\begin{align*}\abs{\posV{{v_i}}{0} - \posV{{v_j}}{0}} &\geq \posV{{v_2}}{0} - \posV{{v_0}}{0}\\ &= \eps - 3 (n-2)\hat{m} + \eps - 3 (n-1)\hat{m}\\ &= 2\eps - 3(2n-3)\eps/(n^2 + 5n -1)\\ &> \eps
\end{align*}for all $0 \leq i,j \leq 2n$ with $\abs{i-j} \geq 2$ and therefore the above-mentioned edges are the only ones in $\CG{0}$.

We proceed by verifying that for each $v \in V$ it holds that $\abs{\moveV[0]{v}} = \hat{m}$. We calculate the available movement for $\ell$. Let $v \in C_\ell$.
Since all $n$ agents in $C_\ell$ have the same initial position $\posV{v}{0} = 0$, it holds that
\begin{align*}
\abs{\moveV[0]{\ell}}  &= \frac{n\cdot \posV{v}{0} + \posV[1]{v_1}{0} - (n+1) \posV{\ell}{0}}{n+2} \\
        &= \frac{\posV{\ell}{0}+ \eps - 3 (n-1)\hat{m} - (n+1) \posV{\ell}{0}}{n+2} \\
        &= \frac{\eps - 3 (n-1)\hat{m} - n \cdot \posV{\ell}{0}}{n+2} \\
        &= \frac{\eps - 3 (n-1)\hat{m}- n \cdot \hat{m}\cdot (n +1)}{n+2} \\
        &= \frac{\eps - (n^2 +4n-3)\hat{m}}{n+2} \\
        &= \frac{\eps - (n^2 +4n-3) \cdot \eps/(n^2 + 5n -1)}{n+2} \\
        &= \eps \cdot \frac{(n^2 + 5n -1) - (n^2 +4n-3)}{(n^2 + 5n -1)(n+2)} \\
        &= \eps \cdot \frac{n +2}{(n^2 + 5n -1)(n+2)} \\
        &= \hat{m}\;.
\end{align*}
The calculation of the available movement of the other agents is analogous.
Let $v \in C_\ell$, then it holds that 
\begin{align*}
\abs{\moveV[0]{v}} & = \frac{\posV{\ell}{0}- \posV{v}{0}}{n+1} + \sum_{v'\in C_\ell}\frac{\posV{v'}{0}- \posV{v}{0}}{n+1}\\
 & = \frac{\posV{\ell}{0}}{n+1} = \frac{\hat{m}\cdot (n + 1)}{n+1} = \hat{m}.
\end{align*}
Let $v \in C_r$, then we have 
\begin{align*}
\abs{\moveV[0]{v}} & = \left\lvert\frac{\posV{r}{0}- \posV{v}{0}}{n+1} + \sum_{v'\in C_\ell}\frac{\posV{v'}{0}- \posV{v}{0}}{n+1}\right\rvert \\
 & = \left\lvert\frac{\posV{r}{0}- (\posV{r}{0} + (n +1)\hat{m})}{n+1}\right\rvert  \\
 & = \frac{\hat{m}\cdot (n + 1)}{n+1} = \hat{m}.
\end{align*}
Let $v_i \in P$ with $i\leq n-1$, then we get that
\begin{align*}
 \abs{\moveV[0]{v_i}} & = \frac{\abs{\posV{v_{i-1}}{0} + \posV{v_{i+1}}{0}- 2\posV{v_i}{0}}}{3}\\
 & = \frac{\abs{\posV{v_{i-1}}{0} +  \eps - 3 (n-(i+1))\hat{m}- \posV{v_i}{0}}}{3} \\
 & = \frac{\abs{\eps - 3 (n-i)\hat{m}- (\eps - 3 (n-(i+1))\hat{m})}}{3} \\
 & = \frac{3\hat{m}}{3} = \hat{m}.
\end{align*}
Let $v_n \in P$, then it holds that 
\begin{align*}
 \abs{\moveV[0]{v_n}} & = \frac{\abs{\posV{v_{n-1}}{0} + \posV{v_{n+1}}{0}- 2\posV{v_n}{0}}}{3}\\
 & = \frac{\abs{\posV{v_{n-1}}{0} +  \eps - 3 (n+1-n)\hat{m}- \posV{v_n}{0}}}{3} \\
 & = \frac{\abs{\eps - 3 (n+1-n)\hat{m}- (\eps - 3 (n-n)\hat{m})}}{3} \\
 & = \frac{3\hat{m}}{3} = \hat{m}.
\end{align*}
Let $v_i \in P$ with $i>n$, then we have 
\begin{align*}
 \abs{\moveV[0]{v_i}} & = \frac{\abs{\posV{v_{i-1}}{0} + \posV{v_{i+1}}{0}- 2\posV{v_i}{0}}}{3}\\
 & = \frac{\abs{\posV{v_{i-1}}{0} +  \eps - 3 (i+1-n)\hat{m}- \posV{v_i}{0}}}{3} \\
 & = \frac{\abs{\eps - 3 (i+1-n)\hat{m}- (\eps - 3 (i-n)\hat{m})}}{3} \\
 & = \frac{3\hat{m}}{3} = \hat{m}.
\end{align*}
Finally, we get that
\begin{align*}
\abs{\moveV[0]{r}}  &=\left\rvert \sum_{v \in C_r\cup {v_{2n-2}}} \frac{\posV{v}{0}-\posV{r}{0}}{n+2}\right\rvert\\
        &= \frac{(n+1) \posV{r}{0}}{n+2} \\
        &\quad - \frac{n\cdot (\posV{r}{0} + (n +1)\hat{m})}{n+2} \\
        &\quad -\frac{(\posV{r}{0} - \eps + 3 (n-1)\hat{m})}{n+2} \\
        &= \frac{\eps - 3 (n-1)\hat{m} - n(n +1)\hat{m}}{n+2}\\
        &= \frac{\eps - (n^2 +4n-3)\hat{m}}{n+2} \\
        &= \frac{\eps - (n^2 +4n-3) \cdot \eps/(n^2 + 5n -1)}{n+2} \\
        &= \eps \cdot \frac{(n^2 + 5n -1) - (n^2 +4n-3)}{(n^2 + 5n -1)(n+2)} \\
        &= \eps \cdot \frac{n +2}{(n^2 + 5n -1)(n+2)} \\
        &= \hat{m}\;.
\end{align*}

By \Cref{lem:etesami_modified}, the expected potential drop is given by
\begin{align*}
    &\E[\Phi(S_0) - \Phi(S_1)] \\
    & = \frac{1}{4n} \sum_{v \in V}(\abs{\Nv[0]{v}}+1) \cdot \abs{\moveV[0]{v}}^2\\
    & = \frac{1}{4n} \left(\frac{n}{2} \left(\frac{n}{4}+2\right) + 2\left(\frac{n}{4}+3\right) + \left(\frac{n}{2}-2\right)\cdot 4\right)\hat{m}^2\\
    & = \frac{1}{4}\left(n/8+7/2 -2/n\right) \hat{m}^2\\
    & = \frac{1}{4} \left(n/8+7/2 - 2/n\right)  (\eps/(n^2/16 +5n/4-1))^2\\
    & = \Theta(\eps^2/n^3).
\end{align*}
On the other hand, there exist $\frac{n}{2}\left(\frac{n}{2}-2\right)$ edges with length longer than $\eps$ and hence $\Phi(S_0) = \Theta(\eps^2n^2)$.
\end{proof}
Note that this only proves that the drop in the first step is sufficiently small.
The expected drop for the next step could increase after activating a node. 
Therefore, this theorem does not prove a lower bound for the convergence time. 
Instead, it shows that the analysis of the step-by-step drop is tight.

In the next section, we see an example of how to change the analysis to circumvent this bound.

\section{Special Network Topologies}
In this section, we will prove two improved upper bounds, each for a more restricted set of graph classes. 
The first result holds when the social network is a complete graph, while the second holds when, in each step of the HKS, the influence network is the same as the social network.

To prove the result for HKSs with a social network, we prove the following characteristics of these systems.
\begin{lemma}
\label{clm:conectedComponent}
Let $\HKS$ be any instance of a $d$-dimen\-sional HKS with complete social network $G=K_n$ and current influence network $\CG{ }$. If all edges in $\CG{ }$ have a length of at most $\eps/2$, each connected component in $\CG{ }$ is a complete graph.
\end{lemma}
\begin{proof}
Let $V' \subseteq V$ be the set of nodes of a connected component in $I$.
Assume that $v,u \in V'$, but $\{v,u\} \not \in E(I)$.
Then there exists a shortest path $P = (v,w_1,\dots,w_k=u)$ of length at least $2$ from $v$ to $u$ where each edge has a length of at most $\eps/2$. 
As a consequence, the distance between $v$ and $w_2$ can be at most $\eps$. 
Therefore, the edge between $v$ and $w_2$ has to exist in the influence network. 
Hence, $P$ is not the shortest path, contradicting the assumption. 
\end{proof}

\UpperBoundComplete*

\begin{proof}
We split this proof into two steps.
First, we count the number of possible steps where the influence network has an edge of length at least $\eps /2$. 
Secondly, we upper-bound the number of steps where the longest edge of the influence network is in $[\delta,\eps/2]$.

Assume in step $t$ there is an edge in the influence network with length at least $\eps /2$. Let $S_t$ and $S_{t-1}$ denote the states of the HKS in steps $t$ and $t-1$, respectively. 
In this case, by \cref{lma:ExpPotentialDrop}, we have
\[\E[\Phi(S_t)-\Phi(S_{t-1})] \geq \frac{\eps^2}{2n \abs{\Et{t}}}\;.\]
As a consequence, the expected number of such steps is bounded by $\abs{E}^2n = \Oh(n^5)$.

For the rest of the proof, assume that all edges in the influence network are shorter than $\eps/2$ and there exists one edge with a length at least $\delta$.
We project the HKS to one dimension along the longest edge. 
By \cref{lem:project}, we know that in the projected graph, no edge increases its length, and there exists an edge with a length at least $\delta$.

Let there be $k$ connected components $C_i = (V_i,E_i)$, $i \in \{1,\dots,k\}$, and $\lambda_i(t)$ be the length of the longest edge in the connected component $C_i$. We bound the total available movement in this component from below using \cref{lem:summedMovement}.

For each connected component $C_i$ with $\lambda_i(t)> 0$ let $e_i = \{u,w\}$ be a longest edge in this component.
We partition $V_i$ into $V_{i,\ell}$ and $V_{i,r}$ at $c = (\posV{u}{t}+\posV{w}{t})/2$ and we define the set $E_{\ell,r,i}$ as in \cref{lem:summedMovement}.
Since the connected component is the complete graph by \cref{clm:conectedComponent} each node from $V_{i,\ell}$ is connected to $w$ while each node from $V_{i,r}$ is connected to $u$. As a consequence, the set $E_{\ell,r,i}$ contains at least $(\abs{V_i}-1)$ edges of length at least $\lambda_t/2$ and one of them has length $\lambda_t$.
As a consequence, $\sum_{e \in E_{\ell,r,i}} \|x_t(e)\|_2 \geq \abs{V_i}\lambda_t/2$ and hence, by \cref{lem:summedMovement},
\begin{align*}
    \abs{V_i} \sum_{v \in V_i} \abs{\moveV[t]{v}}  &= \sum_{v \in V_{i}} \abs{\Nv[t]{v}}\abs{\moveV[t]{v}} \\
    &\geq 2 \sum_{e \in E_{\ell,r,i}} \|x_t(e)\|_2\geq  \abs{V_i}\lambda_i(t)
\end{align*}
and therefore 
\[ 
    \sum_{v \in V_i} \abs{\moveV[t]{v}} \geq \lambda_i(t)\;.
\]
As a consequence, it holds that
\begin{align*}
    &\E[\Phi(S_t)-\Phi(S_{t-1})] \\
    &\geq \frac{1}{n}\sum_{v \in V} (\abs{\Nv[t]{v}} + 1)\|\moveV[t]{v}\|_2^2\\
    &\geq \frac{1}{n}\sum_{i = 1}^k (\abs{V_i}+1)\sum_{v \in V_i}\|\moveV[t]{v}\|_2^2\\
    & \geq \frac{1}{n}\sum_{i = 1}^k (\abs{V_i}+1) \left(\sum_{v \in V_i}\|\moveV[t]{v}\|_2\right)^2/\abs{V_i}\\
    & > \frac{1}{n} \cdot \sum_{i = 1}^k (\lambda_{i}(t))^2.
\end{align*}
Since one of the edges $\lambda_{i}(t)$ has a length at least $\delta$, the expected potential drop is at least $\delta^2/n$.
Therefore, in expectation, there are at most $\Oh(\abs{E} n (\eps/\delta)^2)$ steps where the length of the longest edge is in $[\delta,\eps/2]$.
Combining the two results finishes the proof.
\end{proof}

\noindent We say an HKS is \emph{socially stable} if, independently of the update steps, the influence network is always equal to the social network.
For these systems, we can prove a better upper bound on the expected number of steps needed to reach a $\delta$-stable state.
Examples of such graphs are the path, where all the nodes are positioned with equal distance of at most $\eps$, as well as the graph from \cref{rem:RemakTightniss} if the social network for the latter is reduced to the set of edges in $\Et{0}$.

\begin{theorem}
Let $\HKS$ be a HKS where the social network and the influence network are equal in each step.
Using uniform asynchronous update steps, the expected convergence time to a $\delta$-stable state is bounded by $\Oh(n \abs{E}^2 \log(\eps/\delta))$.
\end{theorem}
\begin{proof}
Note that at any step, it holds that $\Phi(S_t) \leq \abs{E}(\lambda_t)^2$, where $\lambda_t$ is the length of the longest edge at time $t$.
By \cref{lma:ExpPotentialDrop}, the expected drop of the potential in each step is bounded by ${2(\lambda_t)^2}/({n\abs{E}})$.
As a consequence, for each $i \in \mathbb{N}$ the expected number of steps with $\lambda_t \in [\eps/2^{i+1},\eps/2^{i}]$ is bounded by $\Oh(n\abs{E}^2)$.
Since for $\lambda_t \in [\delta,\eps]$ there are at most $\log(\eps/\delta)$ such intervals, the expected number of update steps is bounded by $\Oh(n \abs{E}^2 \log(\eps/\delta))$.
\end{proof}

\section{Lower Bound}
In this section, we complement our upper bounds on the expected convergence time with a lower bound. 
To the best of our knowledge, no lower bound for asynchronous updates is known so far.
We prove that there exists a family of instances of HKS $(S_{0,4n})_{n \in \mathbb{N}}$ for which at least $\Omega(\Phi(S_{0,4n}) n \abs{E})$ updates are needed to converge. 
In this family, we have $\abs{E} = \Theta(n^2)$ and $\Phi(S_0) = \Theta(n)$ for $\eps=1$ and $\delta < 1/2$.
Note that the lower bound holds for any given update sequence. 
In particular, we also prove that there cannot be a deterministic algorithm that reaches a $\delta$-stable state faster than the proven lower bound.

\thmLowerBound*

Note that this lower bound is tight with regard to the considered family of instances and the parameter $n$ since \cref{thm:upper-bound} states that in expectation at most $\Oh((\Phi(S_{0,4n}) n \abs{E})/\delta^2) = \Oh(n^4)$ updates are needed, to reach a $\delta$-stable state.

We prove this lower bound in three steps. First, we prove that edges cannot be deactivated in the defined family of instances. 
Then, we upper bound the total available movement of certain nodes in the family by a modified process that is simpler to analyze.

\subparagraph*{Dumbbell Graph}
We define the following family of instances of HKS $(S_{0,4n})_{n \in \mathbb{N}} = (G_{4n} =(V,E),\eps,x_0))_{n \in \mathbb{N}}$: Let $G_{4n}= (V,E)$ be the Dumbbell graph defined as follows. The $\abs{V} = 4n$ nodes of the graph are partitioned into the sets $C_{\ell },C_{r}, P,\{\ell,r\} \subseteq V$, such that $\abs{C_{\ell}} = \abs{C_{r}} =n $, $\abs{P} = 2n-2$.

The set of edges are given such that $C_{\ell}\cup \{\ell\}$ and $C_{r}\cup \{r\}$ are cliques, $P = \{p_1,\dots,p_{2n-2}\}$ is a path such that $\{p_i,p_{i+1}\} \in E$ for each $i \in \{1,\dots,2n-2\}$, as well as $\{\ell,p_1\},\{r,p_{2n-2}\} \in E$.
In state $S_{0,4n}$ the opinions of the agents that correspond to the nodes $V$ are given as follows:
\begin{itemize}
    \item $x_0(v) = 0$ for each $v \in C_{\ell}$
    \item $x_0(\ell) = \eps/n$
    \item $x_0(p_i) = i \eps + \eps/n$ for each node $p_i \in P$
    \item $x_0(r) = (2n-1)\eps + \eps/n$
    \item $x_0(v) = (2n-1)\eps + 2\eps/n$ for each $v \in C_{r}$.
\end{itemize}
Note that $\Phi(S_{0,4n}) \in \Omega(n\eps^2)$ and $\abs{E} = n^2$, therefore $\Omega(\Phi(S_0)n \abs{E}/\delta^2) = \Omega(n^4)$.

\begin{lemma}
\label{lma:deactivatingEdges}
When starting the HKS process with $S_{0,4n}$, it can never happen that an edge is deactivated during the process.
\end{lemma}
\begin{proof}
Assume we activate node $p_i \in P$ in step $t+1$ and that the edges $\{p_{i-1},p_i\}$ and $\{p_{i},p_{i+1}\}$ are active as well as $x_{t}(p_{i-1})\leq x_{t}(p_{i})\leq x_{t}(p_{i+1})$. Which means that $x_t(p_{i}) - x_t(p_{i-1}) \leq \eps$ and $x_t(p_{i+1}) - x_t(p_i) \leq \eps$ and hence $\abs{x_t(p_{i+1}) - x_t(p_{i-1})} \leq 2\eps$.
As a consequence,
\begin{align*}
    x&_{t+1}(p_{i}) - x_{t+1}(p_{i-1})\\
    &= \frac{1}{3}(x_{t}(p_{i})+x_{t}(p_{i+1})+x_{t}(p_{i-1})) - x_t(p_{i-1})\\
    &= \frac{1}{3}((x_t(p_{i}) - x_{t}(p_{i-1})) + (x_t(p_{i+1}) - x_{t}(p_{i-1})))\\
    & \leq \frac{1}{3}(\eps +2\eps) = \eps
\end{align*}
and similarly $x_t(p_{i+1}) - x_{t+1}(p_i) \leq \eps$. Furthermore, we obviously have that $x_{t}(p_{i-1}) \leq \frac{1}{3}(x_{t}(p_{i})+x_{t}(p_{i+1})+x_{t}(p_{i-1})) \leq x_{t}(p_{i+1})$ and hence  $x_{t+1}(p_{i-1}) \leq x_{t+1}(p_{i}) \leq x_{t+1}(p_{i+1})$ 

For each pair of nodes $v,v' \in C_{\ell}$ and any time $t$, we will prove that $\abs{x_t(v')- x_{t}(v)} \leq \eps/n$, that $\abs{x_{t}(\ell) - x_{t}(v)} \leq 2\eps/n$, and that $\abs{x_{t}(\ell) - x_{t}(p_1)} \leq \eps$. This claim is true in the start configuration since all nodes in $C_{\ell}$ have the same position, and $\ell$ has distance $\eps/n$ to the nodes in $C_{\ell}$.

Assume we activate a node $v \in C_{\ell}$ in step $t+1$ and $\abs{x_t(v')- x_{t}(v'')} \leq \eps/n$ as well as $\abs{x_{t}(\ell) - x_{t}(v')} \leq 2\eps/n$ for each pair of nodes $v',v'' \in C_{\ell} \cup \{\ell\}$. Then it holds for any node $v' \in C_{\ell} \cup \{\ell\} \setminus \{v\}$ that
\begin{align*}
    \abs{&x_{t}(v') - x_{t+1}(v)} \\
    & = \abs{x_{t}(v') - \sum_{v'' \in C_{\ell} \cup \{\ell\}}x_{t}(v'')/(\abs{C_{\ell} \cup \{\ell\}})}\\
    & \leq \sum_{v'' \in C_{\ell} \cup \{\ell\} \setminus \{v'\}} \abs{x_{t}(v') - x_{t}(v'')}/(n+1)\\
    & \leq ((n-1) (\eps/n) + 2\eps/n)/(n+1) = \eps/n\;.
\end{align*}
Furthermore, it holds that
\begin{align*}
    \abs{&x_{t}(\ell) - x_{t+1}(v)} \\
    & = \abs{x_{t}(\ell) - \sum_{v'' \in C_{\ell} \cup \{\ell\}}x_{t}(v'')/(\abs{C_{\ell} \cup \{\ell\}})}\\
    & \leq \sum_{v'' \in C_{\ell}} \abs{x_{t}(v') - x_{t}(v'')}/(n+1)\\
    & \leq (n \cdot (\eps/n))/(n+1) \leq 2\eps/n\;.
\end{align*}

Assume we activate the node $\ell$ in step $t+1$ and for each pair of nodes $v,v' \in C_{\ell}$ it holds that $\abs{x_t(v')- x_{t}(v)} \leq \eps/n$, that $\abs{x_{t}(\ell) - x_{t}(v)} \leq 2\eps/n$, and that $\abs{x_{t}(\ell) - x_{t}(p_1)} \leq \eps$. Then it holds for any node $v' \in C_{\ell}$ that
\begin{align*}
    \abs{&x_{t}(v') - x_{t+1}(\ell)} \\
    & = \abs{x_{t}(v') - \sum_{v'' \in C_{\ell} \cup \{\ell,p_1\}}x_{t}(v'')/(\abs{C_{\ell} \cup \{\ell,p_1\}})}\\
    & \leq \sum_{v'' \in C_{\ell} \cup \{\ell,p_1\}\setminus \{v'\}} \abs{x_{t}(v') - x_{t}(v'')}/(n+2)\\
    & \leq ((n-1) (\eps/n)+ 2\eps/n + \eps)/(n+2)\\
    & = (2\eps+ \eps/n)/(n+2) \leq 2\eps/n
\end{align*}
Furthermore, it holds that
\begin{align*}
    \abs{&x_{t}(p_1) - x_{t+1}(\ell)} \\
    & = \abs{x_{t}(p_1) - \sum_{v'' \in C_{\ell} \cup \{\ell,p_1\}}x_{t}(v'')/(\abs{C_{\ell} \cup \{\ell,p_1\}})}\\
    & \leq \sum_{v'' \in C_{\ell} \cup \{\ell\}} \abs{x_{t}(p_1) - x_{t}(v'')}/(n+2)\\
    & \leq (n (\eps/n+\eps)+ \eps)/(n+2) = \eps.
\end{align*}
Analogously, we can prove the statement for the nodes in $C_{r} \cup \{r\}$.

Hence, no edge can disappear when updating any node in the graph at any step $t$.
\end{proof}

A necessary condition to terminate is that each edge on the path $P$ has a length of at most $\delta$. 
Since the distance between $\ell$ and $r$ is $(2n-1)\eps$ at state $S_{0,4n}$ and can be at most $(2n-1)\delta$ in the final state, one of the nodes $\ell$ and $r$ has to move by at least $(\eps-\delta)(2n-1)/2$.

\subparagraph*{Modified Update Process}
In the following, we will lower-bound the number of required update steps to move $\ell$ by this distance by defining a modified update process that moves $\ell$ faster.
However, the required number of steps in the modified process can be analyzed more easily. 
The opinions in the modified process are denoted by $x'_t(v)$ for each $v \in V$ and $t \in \mathbb{N}$.
Let $S_t = \sum_{v \in C_{\ell} \cup \ell} x_{t}(v)/(n+1)$ and $S'_t = \sum_{v \in C_{\ell} \cup \ell} x'_{t}(v)/(n+1)$ denote the mass center of the nodes $C_{\ell} \cup \ell$ in the original and modified process respectively. 
Note that when activating an agent from $C_{\ell}$ in each of the two processes, it will move to the corresponding center of mass.

Consider any update-sequence $(u_t)_{t \in \mathbb{N}_{\geq 0}}$ specifying for any $t \in \mathbb{N}_{\geq 0}$ the node $u_t$ that will be updated to generate the positions $x_{t+1}$ from $x_t$.
Assume for simplicity of notation, that the nodes in $C_{\ell}$ are numbered from one to $n$, i.e., $C_{\ell} = \{v_1,\dots,v_n\}$.
Given $(u_t)_{t \in \mathbb{N}_{\geq 0}}$, the modified process is defined as follows: 
\begin{itemize}
    \item $x'_0 = x_0$.
    \item If $u_{t} \not \in C_{\ell}$ do nothing, i.e., set $x'_{t+1} = x'_t$.
    \item If $u_{t} \in C_{\ell}$, update multiple nodes at once. Generate $x'_{t+1}$ as follows:
    \begin{itemize}
        \item find the node $v_{\min} \in C_{\ell}$ that has the smallest value $x'_{t}(v)$ with the smallest index, and set 
        \[x'_{t+1}(v_{\min}) = S'_t\;;\]
        \item move the node $\ell$ to the right:
        set \[x'_{t+1}(\ell) = \frac{\eps}{n}+\sum_{v \in C_{\ell}}\frac{x'_{t+1}(v)}{n}\;;\]
        \item if $v_{\min} = v_n$, move all the nodes in $C_{\ell}$ to $S'_t$, i.e., for all $v \in C_{\ell}$ set \[x'_{t+1}(v) = S'_t\;.\] 
    \end{itemize}
\end{itemize}
We prove \cref{thm:lower-bound} in two steps. 
First, we prove that, indeed, the modified process moves the node $\ell$ faster than the original process. 
Afterward, we prove that the modified process needs at least $\Omega(n^4)$ steps to move the clique by the distance $(\eps - \delta)(2n-1)/2$.

\begin{lemma}
\label{lma:processDominance}
For each update sequence $(u_t)_{t \in \mathbb{N}_{>0}}$ there exists for each $t \in \mathbb{N}$ a bijection $f_t: C_{\ell} \to C_{\ell}$ such that $x_t(v) \leq x'_t(f(v))$ for each $v \in C_{\ell}$.
Furthermore in each step it holds that $x_t(\ell) \leq x'_t(\ell)$ and $S_t \leq S'_t$.
\end{lemma}
\begin{proof}
First, note that the modified update process will update the nodes $C_{\ell}$ always in the same order from $v_1$ to $v_n$ since it chooses the node with a smallest opinion and the smallest index. 
Since, in the beginning, all the nodes have the same opinion, and the opinion only increases when updated, the first $n$ updates will update the nodes in the claimed order. 
After updating the node $v_n$, all the nodes will be shifted to the same (increased) opinion. 
Inductively, the nodes are activated in the same order.

We will prove the claim via induction on the number of steps.
Before activating any node, the bijection $f_0: C_{\ell} \to C_{\ell}, v \mapsto v$ fulfills the required properties, since $x_0(v)= x'_0(v)=0$.
Furthermore it holds that $x_0(\ell) = x'_0(\ell) = \eps/n$ and \[S_0 = S'_0 = \sum_{v \in C_{\ell} \cup \ell} x_{0}(v)/(n+1) =\frac{\eps}{n(n+1)}\;.\]

Assume that after step $t$ there is a bijection $f_t: C_{\ell} \to C_{\ell}$ such that $x_t(v) \leq x'_t(f_t(v))$ for each $v \in C_{\ell}$, $S_t \leq S'_t$, and $x_t(\ell) \leq x'_t(\ell)$.
We have to consider the following cases: the updated node is in $C_{\ell}$, the updated node is $\ell$, the modified process updates the node $v_n$, and a node not in $C_{\ell} \cup \{\ell\}$ is updated.

In the case that a node in $V \setminus (C_{\ell} \cup \{\ell\})$ is updated, the values $x_t(\ell), x'_t(\ell), S_t, S'_t$ stay unchanged. 
Furthermore, in both processes, no node from $C_{\ell}$ will be updated and hence we set the update function $f_{t+1}= f_t$.

If the node $\ell$ is updated, the modified process will not move the node $\ell$.
Note that in the original process, the edge connecting the node $\ell$ with the node $p_1$ can have length at most $\eps$, since no edge is deactivated when updating any node and hence $x_t(p_1) \leq x_t(\ell) + \eps$.
As a consequence, it holds that 
\begin{align*}
    x_{t+1}(\ell) &= \sum_{v \in C_{\ell} \cup \{\ell,p_1\}} x_t(v)/(n+2)\\
    &\leq ((n+1) S_t + (x_t(\ell) + \eps))/(n+2)\\
    &\leq ((n+1) S'_t + (x'_t(\ell) + \eps))/(n+2)\\
    &= (\sum_{v \in C_{\ell}} x'_{t}(v) + 2x'_t(\ell) +\eps)/(n+2)\\
    &= \frac{(n+2)\sum_{v \in C_{\ell}} x'_{t}(v)/n  + (n+2)\eps/n}{(n+2)}\\
    & = \sum_{v \in C_{\ell}}x'_{t}(v)/n + \eps/n = x'_t(\ell)
\end{align*}
and hence $x_{t+1}(\ell) \leq x'_{t+1}(\ell)$.
Since in both processes no node form $C_{\ell}$ was updated, for the bijection $f_t: C_{\ell} \to C_{\ell}$ with $x_{t}(v) \leq x'_{t}(f(v))$ for each $v \in C_{\ell}$, it still hods that $x_{t+1}(v) \leq x'_{t+1}(f(v))$ and hence we can set $f_{t+1} = f_t$. 
As a consequence $S_{t+1} \leq S'_{t+1}$ has to hold as well.

Let us assume that we update the clique node $v_{t+1} \in C_{\ell}$ in the next step in the original process and the node $v'_{t+1} \in C_{\ell}$ in the modified process.
Note that $x_{t+1}(v_{t+1}) = S_t$ and $x'_{t+1}(v'_{t+1}) = S'_t$ by definition of the update process.
Furthermore 
\[
    S_{t+1} = S_t +(x_{t+1}(v_t) - x_{t}(v_t))/(n+1).
\]

If $f_{t}(v_{t+1}) = v'_{t+1}$, we know that $x_{t+1}(v_{t+1}) = S_t \leq S'_t = x'_{t+1}(v'_{t+1}) = x'_{t+1}(f_{t}(v_{t+1}))$.
Since no other node form $C_{\ell}$ is moved, we define $f_{t+1} = f_{t}$ and it holds that $x_{t+1}(v)\leq x'_{t+1}(f(v))$ for each $v \in C_{\ell}$. Furthermore,
\begin{align*}
    S_{t+1} &= \sum_{v \in C_{\ell} \cup \{\ell\}} x_{t+1}(v)/(n+1)\\
    & \leq x'_{t+1}(\ell)/(n+1) + \sum_{v \in C_{\ell}} x'_{t+1}(f_{t+1}(v))/(n+1)\\
    & = \sum_{v \in C_{\ell} \cup \{\ell\}} x'_{t+1}(v)/(n+1) \\
    & = S'_{t+1}.
\end{align*}

If $f_{t}(v_{t+1}) \not = v'_{t+1}$, consider the nodes $f_{t}(v_{t+1})$ and $f^{-1}_{t}(v'_{t+1})$. We know that $x_{t}(v_{t+1}) \leq x'_{t}(f_t(v_{t+1}))$ and $x_{t}(f^{-1}_{t}(v'_{t+1})) \leq x'_{v'_{t+1}}$. Furthermore, by the choice of $v'_{t+1}$, we know that $ x'_{v'_{t+1}}\leq x'_{t}(f_t(v_{t+1}))$. As a consequence, it holds that $x_t(f^{-1}_{t}(v'_{t+1})) \leq x'_t(f_t(v_{t+1}))$. When updating the nodes, we get $x_{t+1}(v_{t+1}) = S_t \leq S'_t = x'_{t+1}(v'_{t+1})$. We define
\begin{align*}
    f_{t+1}(v) = 
    \begin{cases}
    v'_{t+1} & \text{ if } v = v_{t+1}\\
    f_{t}(v_{t+1}) & \text{ if } v = f^{-1}_{t}(v'_{t+1})\\
    f_{t}(v) & \text{ otherwise}
    \end{cases}
\end{align*}
Note that for each $v \in C_{\ell}$ it holds that $x_{t+1}(v) \leq x'_{t+1}(f_{t+1}(v))$. Furthermore, we have 
\[
    x_{t+1}(\ell) = x_{t}(\ell) \leq x'_{t}(\ell) \leq x'_{t+1}(\ell)
\]
and 
\begin{align*}
S_{t+1} &= \sum_{v \in C_{\ell} \cup \{\ell\}} x_{t+1}(v)/(n+1) \\
& \leq (x'_{t+1}(\ell) + \sum_{v \in C_{\ell}} x'_{t+1}(f_{t+1}(v)))/(n+1)\\ &
= S'_{t+1}.
\end{align*}

Finally, if the updated node in the modified processes was $v_n$, all the nodes in $C_{\ell}$ move to the point $S'_{t}$. 
Since $S_t \leq S'_t$, the identity function $f_{t+1} = id$ fulfills the required conditions, since $S'_t$ is monotonically increasing and hence there can be no node with $x_t(v) > S'_t$.
Similarly as above it follows that $x_{t+1}(\ell) \leq x'_{t+1}(\ell)$ and $S_{t+1} \leq S'_{t+1}$.
\end{proof}

\begin{lemma}
\label{lma:convergencModifiedProcess}
The modified process needs at least $\Omega(n^4)$ updates before $\ell$ has moved by at least $(\eps - \delta)(2n-1)/2$.
\end{lemma}
\begin{proof}
In the modified process, the nodes in $C_{\ell}$ are shifted to a common position every $n$th update step.
We define an update round as $n$ updates of nodes in $C_{\ell}$ such that, before the first update of the round, all nodes have the same position, and all nodes are shifted to the same position after the last update of the round.

First, we prove that in each update round, the node $\ell$ moves by at most $\frac{(1+e)\eps}{n(n+1)}$. Let $p_0$ be the common start position of all nodes in $C_{\ell}$. By definition of the modified process, the node $\ell$ has position $p_0 + \eps/n$.
We denote by $S'_t = \sum_{v \in C_{\ell} \cup \ell} x'_{t}/(n+1)$ the mass center of the nodes $C_{\ell} \cup \ell$ and by $\bar{S}'_t = \sum_{v \in C_{\ell}} x'_{t}/n$ 
the mass-center of the clique nodes $C_{\ell}$. 

When updating a node $u_{t+1} \in C_{\ell}$ at step $t+1$, it moves from position $p_0$ to position $S'_{t}$ since in each round, only non-updated nodes will be moved. This update increases the mass-center of the clique $\bar{S}'_{t}$ to  
\[\bar{S}'_{t+1} = \bar{S}'_{t} + (S'_{t} - p_0)/n.\]
After updating the node $u_{t+1}$, the node $\ell$ is shifted to $x'_t(\ell) = \bar{S}'_{t+1} + \eps/n$,
increasing the mass-center of the nodes $C_{\ell} \cup \ell$ from $S_{t}$ to 
\begin{align*}
S'_{t+1} &= (\bar{S}'_{t+1} \cdot n + \bar{S}'_{t+1} + \eps/n)/(n+1) \\
&= \bar{S}'_{t+1} + \frac{\eps}{n(n+1)}.
\end{align*}
As a consequence, we get that 
\begin{align*}
\bar{S}'_{t+1} &= \bar{S}'_{t} + (S'_{t} - p_0)/n\\
&= \bar{S}'_{t} + \frac{\bar{S}'_{t} + \frac{\eps}{n(n+1)} - \bar{S}'_0}{n}\\
&= \frac{n+1}{n}\bar{S}'_{t} - \frac{\bar{S}'_0}{n} + \frac{\eps}{n^2(n+1)} \\
\end{align*}
Via induction, we show that 
\[
    \bar{S}'_{t} = \bar{S}'_0 + \left(\left(\frac{n+1}{n}\right)^{t}-1\right) \cdot \frac{\eps}{n(n+1)}
\]
It holds that
\begin{align*}
    \bar{S}'_{1} & = \frac{n+1}{n}\bar{S}'_{0} - \frac{\bar{S}'_0}{n} + \frac{\eps}{n^2(n+1)}\\
    &= \bar{S}'_{0} + \frac{\eps}{n^2(n+1)}\\
    &= \bar{S}'_0 + \left(\left(\frac{n+1}{n}\right)^{1}-1\right) \cdot \frac{\eps}{n(n+1)}
\end{align*}
Additionally, it holds that
\begin{align*}
  \bar{S}'_{t+1} =& \frac{n+1}{n}\bar{S}'_{t} - \frac{\bar{S}'_0}{n} + \frac{\eps}{n^2(n+1)}\\
  =& \frac{n+1}{n}\left(\bar{S}'_0 + \left(\left(\frac{n+1}{n}\right)^{t}-1\right) \cdot \frac{\eps}{n(n+1)}\right)\\
  &- \frac{\bar{S}'_0}{n} + \frac{\eps}{n^2(n+1)}\\
  = & \bar{S}'_0 + \left(\left(\frac{n+1}{n}\right)^{t+1}- \frac{n+1}{n} +\frac{1}{n} \right)\frac{\eps}{n(n+1)}\\
  = & \bar{S}'_0 + \left(\left(\frac{n+1}{n}\right)^{t+1} -1 \right) \cdot \frac{\eps}{n(n+1)}
\end{align*}
As a consequence,
\begin{align*}
\bar{S}'_{n-1} &= \bar{S}'_0 + \left(\left(\frac{n+1}{n}\right)^{n-1} -1 \right) \cdot \frac{\eps}{n(n+1)}\\
& <  \bar{S}'_0 + \frac{(e-1) \eps}{n(n+1)}
\end{align*}
In the last update step of the round, all nodes are moved to 
\[
    S'_{n-1} = \bar{S}'_{n-1} + \frac{\eps}{n(n+1)} \leq \bar{S}'_0 + \frac{ e\eps}{n(n+1)}
\]
resulting in $\bar{S}'_n \leq \bar{S}'_0 + \frac{ e \eps}{n(n+1)}$ 
Consequently $\ell$ is located at or left of position $\bar{S}'_0 + \frac{e \eps}{n(n+1)} + \frac{\eps}{n}$ and hence has moved by at most $\frac{e\eps}{n(n+1)}$ in this round.

Note that after moving each node once, the relative positioning between the nodes in $C_{\ell} \cup \ell$ is the same as before moving the first node. Hence, it is sufficient to consider the total number of rounds of moving $n$ nodes from the clique.

The node $\ell$ has to move by $(\eps - \delta)(n-1)$, but after activating each node of the clique, it has moved by at most $\frac{e \eps}{n(n+1)}$. Hence, we need 
\[
    \frac{(\eps - \delta)(n-1)}{\frac{e \eps}{n(n+1)}} = \Omega(n^3 \frac{\eps -\delta}{\eps}) = \Omega(n^3)
\]
of these rounds, for $\delta$ small enough. Since each round needs $n$ updates, at least $\Omega(n^4)$ updates are needed until the modified process terminates.
\end{proof}

\begin{proof}[Proof of \cref{thm:lower-bound}]
Consider the family $(S_{0,4n})_{n \in \mathbb{N}}$ as defined in this section. 
By \cref{lma:deactivatingEdges}, no edge will be deactivated until the process terminates. 
As a consequence, for termination, each edge on the path can have a length of at most $\delta$.
This implies the the distance between $\ell$ and $r$ has to shrink from $(2n-1) \eps$ to $(2n-1) \delta$.
Therefore, one of the nodes $\ell$ or $r$ must move by at least $(n-1/2)(\eps -\delta)$ until the process terminates.
Let us w.l.o.g. assume that $\ell$ moves by this distance.

Consider the modified process as defined in this chapter.
By \cref{lma:processDominance} we know that $x_t(\ell) \leq x'_t(\ell)$ for any update step $t$. 
Furthermore, by \cref{lma:convergencModifiedProcess} we know, that at least $t \in \Omega(n^4)$ steps are needed before $x'_t(\ell)-x'_0(\ell)\geq (n-1/2)(\eps -\delta)$.
Consequently, the original process needs at least $\Omega(n^4)$ update steps until it terminates. 
\end{proof}

\section{Simulation Results}\label{sec:simulations}
To corroborate our theoretical findings, we performed agent-based simulations of asynchronous Hegselmann-Krause opinion dynamics in one dimension on two types of initial HKS states called \textsl{Path} and \textsl{Dumbbell}. They are defined as follows:
\begin{itemize}
    \item \textsl{Path:} The given social network is a path graph. Initially, the agents' opinions are uniformly distributed in one dimension with an equal distance of $\eps$ so that the influence network forms a path graph with a uniform edge length of $\eps$.
    \item \textsl{Dumbbell:} This is the state constructed in the proof of \cref{rem:RemakTightniss} using the dumbbell graph, except that the social network contains only the edges that are in $\Et{0}$
\end{itemize}
We fixed $\eps = 100$ and $\delta = 1$ in our simulations.
For each initial HKS state on social networks with varying numbers of agents $n$, we simulated 100 independent runs of random activations needed to reach a $\delta$-stable state. The code for our simulator software and all necessary tools to reproduce our figures are available from our public GitHub repository.\footnote{\url{https://github.com/dcmx/HKsim}}

We present our simulation results in \cref{fig:SimulationPathnadDumbbell}.
\begin{figure}[t]
    \centering
    \resizebox{0.45\textwidth}{!}{\input{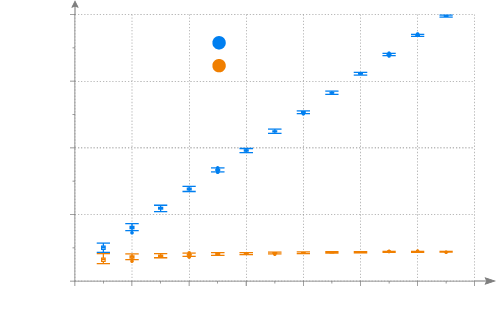}}
    \vspace{0.3cm}

    \caption{%
The plot shows the normalized convergence time: the number of agent activations until a $\delta$-stable state has been reached, divided by $n^3$. The data indicate that the convergence time on \textsl{Path} instances with equal distances scales as $n^3$ and on \textsl{Dumbbell} instances it scales as $n^4$.
    }
    \label{fig:SimulationPathnadDumbbell}
\end{figure}
There, the obtained number of activations divided by $n^3$ is plotted via box plots that summarize the results for each configuration.
Since for \textsl{Path} instances, the number of activations appears to be constant, we observe that we need $\Theta(n^3)$ activations for \textsl{Path} instances.
On the other hand, the number of activations seems to grow linearly in $n$ for \textsl{Dumbbell} instances. 
This matches our proofs (upper and lower bound) that in $\Theta(n^4)$ activations \textsl{Dumbbell} instances reach a $\delta$-stable for constant $\eps$ and $\delta$.

Note that by construction, in the first step, the potential function of both instance types is bounded by $\Phi(S_0) = \Theta(n \eps^2)$. Applying \cref{thm:upper-bound} yields an upper bound of $\Phi(S_0)/(2\delta^2/(n\abs{E})) = \Oh(n \eps^2 / (2\delta^2/(n\abs{E})))$, which yields an upper bound of $\Oh(n^3 (\eps/\delta)^2)$ for \textsl{Path} instances and 
$\Oh(n^4 (\eps/\delta)^2)$ for \textsl{Dumbbell} instances. Thus, the empirically observed lower bounds on the expected number of steps until convergence match our theoretical analysis for these two graph classes with respect to the dependence on the number of agents. 

We highlight that our simulations not only confirm our theoretical results but also allow us to empirically pinpoint the constants hidden in the asymptotic analysis. For the simulations carried out on the path our simulation results indicate a constant of roughly $1.5$, giving a total running time of roughly $1.5 \cdot n^3$. For the dumbbell we get a running time of less than $0.1 \cdot n^4$.
We remark that in our theoretical analysis we did not attempt to optimize the constants. In fact, our result for the dumbbell graph (see the proof of \cref{thm:upper-bound}) gives a bound of at most $\log(100^2)/4 \cdot n^4 \approx 3.32 \cdot n^4$.
Interestingly, the total running times are sharply concentrated around their mean. In fact, the plot in \cref{fig:SimulationPathnadDumbbell} actually shows box plots, but starting with a number of agents as small as only 40 the upper and lower whiskers become almost identical with no outliers detected.

\section{Conclusion}
In this paper, we present the first analysis of the convergence time of asynchronous Hegselmann-Krause opinion dynamics on arbitrary social networks. As our main result, we derive an upper bound of $\Oh(n \abs{E}\Phi(S_0)/\delta^2 ) \leq \Oh(n \abs{E}^2 \left(\eps/\delta \right)^2)$ expected random activations until a $\delta$-stable state is reached. This bound significantly improves over the state-of-the-art upper bound for the special case with a given complete social network. 
Moreover, our simulation results on one-dimensional instances with a path graph or a dumbbell graph as the social network indicate that our theoretical upper bound is tight for these instances. 
For the dumbbell graph, this is not only underlined by our simulations but also proven by presenting a matching lower bound with respect to the number of agents.
This theoretical lower bound on the expected convergence time is the first proven non-trivial lower bound for asynchronous opinion updates.
A challenging open problem is to improve this lower bound by finding a graph class with $\Phi(S_0) \in \Omega(n^2\eps^2)$ with expected convergence time in $\Omega(n^5)$.

It might be possible to prove better bounds for specific social network topologies. 
Regarding this, it would be interesting to consider social networks that have similar features to real-world social networks. 
Moreover, another direction for future work is to consider social networks with directed and possibly weighted edges. 
This would more closely mimic the structure of real-world neighborhood influences, allowing us to study asymmetric influence settings found in online social networks like Twitter. 
Another promising extension would be to incorporate the influence of external factors like publicity campaigns.

\subsection*{Declarations}
This work was supported by DFG Research Group ADYN under grant DFG 411362735.

\balance

\bibliographystyle{ACM-Reference-Format} 
\bibliography{HK}

\end{document}